\documentclass[10pt]{iopart}

\usepackage[utf8]{inputenc}
\usepackage[english]{babel}
\usepackage[colorlinks=true, allcolors=blue]{hyperref}
\usepackage{graphicx}
\usepackage{xcolor}
\usepackage{enumerate}
\usepackage{comment}
\usepackage{csquotes}  
\usepackage[sorting=none]{biblatex}
\expandafter\let\csname equation*\endcsname=\relax 
\expandafter\let\csname endequation*\endcsname=\relax 
\usepackage{amsmath}
\usepackage{amssymb}
\usepackage{amsthm}
\usepackage{mathtools}

\graphicspath{{./figures/}}
\bibliography{references}

\newtheorem{prop}{Proposition}

\newcommand{\abs}[1]{\left\lvert #1 \right\rvert}

\newcommand{\ket}[1]{\left\lvert #1 \right\rangle}
\newcommand{\ketbra}[2]{
    \left\lvert #1 \right\rangle\!
    \left\langle #2 \right\rvert
}

\newcommand{\expval}[1]{\left\langle #1 \right\rangle}

\newcommand{\dd}{\mathrm{d}}
\renewcommand{\Re}[1]{\operatorname{Re}\left\{#1\right\}}
\renewcommand{\Im}[1]{\operatorname{Im}\left\{#1\right\}}
\renewcommand{\tr}[1]{\trace\!\left\{ #1 \right\}}
\DeclareMathOperator{\trace}{tr}
\newcommand{\bound}[2]{E(#1, #2)}

\begin{document}

\title{Quantum speed limit for Kirkwood-Dirac quasiprobabilities}

\author{Sagar Silva Pratapsi\footnote{Corresponding author}}
\ead{spratapsi@uc.pt}
\address{CFisUC, Department of Physics, University of Coimbra, P-3004 - 516 Coimbra, Portugal}
\address{Instituto Superior Técnico, University of Lisbon, 1049-001 Lisbon, Portugal}
\address{Instituto de Telecomunicações, 1049-001 Lisbon, Portugal}

\author{Sebastian Deffner}
\ead{deffner@umbc.edu}
 \address{Department of Physics, University of Maryland, Baltimore County, Baltimore, MD 21250, USA}
 \address{Quantum Science Institute, University of Maryland, Baltimore County, Baltimore, MD 21250, USA}
\address{National Quantum Laboratory, College Park, MD 20740, USA}

\author{Stefano Gherardini}
\ead{stefano.gherardini@ino.cnr.it}
\address{Istituto Nazionale di Ottica del Consiglio Nazionale delle Ricerche (CNR-INO), Largo Enrico Fermi 6, I-50125 Firenze, Italy}
\address{European Laboratory for Non-linear Spectroscopy, Università di Firenze, I-50019 Sesto Fiorentino, Italy}

\begin{abstract}
What is the minimal time until a quantum system undergoing unitary dynamics can exhibit genuine quantum features? To answer this question we derive quantum speed limits for two-time correlation functions arising from statistics of measurements. These two-time correlators are described by Kirkwood-Dirac quasiprobabilities, if the initial quantum state of the system does not commute with the measurement observables. The quantum speed limits here introduced are derived from the Schr\"{o}dinger-Robertson uncertainty relation, and set the minimal time at which the real part of a quasiprobability can become negative and the corresponding imaginary part can be different from zero or crosses a given threshold. This departure of Kirkwood-Dirac quasiprobabilities from positivity is evidence for the onset of non-classical traits in the quantum dynamics. As an illustrative example, we apply these results to a conditional quantum gate by determining the optimal condition that gives rise to non-classicality at maximum speed. In this way, our analysis hints at boosted power extraction due to genuinely non-classical dynamics.  
\end{abstract}

\date{\today}

\maketitle

The question asking whether a quantum process is truly quantum looks as innocuous, as it is deep. While many sophisticated answers could be given, such as referring to violations of Bell~\cite{nielsen2010quantum} or Leggett-Garg~\cite{Leggett1985PRL} inequalities, the simplest answer is arguably found in the presence of non-classical correlations~\cite{Ollivier2001PRA}.

In this paper, we focus on exactly such correlations that characterize the statistics of measuring two distinct quantum observables $A$ and $B$, at the beginning and end of a unitary process with a non-negligible duration. Such two-time correlation functions~\cite{margenau1961correlation} have become ubiquitous in modern physics, ranging from condensed matter physics and quantum chaos~\cite{SilvaPRL2008,chenu2018quantum,dressel2018strengthening,alonso2019out,mohseninia2019strongly,Touil2020QST,Touil2021PRXQ,tripathy2024quantum,carolan2024operator} to quantum thermodynamics~\cite{TalknerPRE2007,Kafri2012PRA,Solinas2016probing,diaz2020quantum,levy2020quasiprobability,maffei2022anomalous,hernandez2022experimental,FrancicaPRE2023,Touil2024EPL}. 
Most quantum correlation functions can be computed even when the exact quantum state is not directly measurable~\cite{Buscemi_2014,SolinasPRAmeasurement,LostaglioKirkwood2022}.

Here we are interested in understanding if any universal statement can be made about the dynamics of such correlators. In fact, their time-evolution, like the dynamics of the system itself, are entirely generated by the Hamiltonian of the system. Hence, it appears obvious to consider \textit{quantum speed limits} for correlation functions~\cite{Pandey2023PRA}~\footnote{Evidently, such limits are no longer useful if the observables $A$ and $B$ are measured consecutively without letting the system evolve.}.
Quantum speed limits (QSLs)~\cite{DeffnerJPA2017} set constraints on the maximum speed of the quantum evolution~\cite{mohan2022quantum,carabba2022quantum}. From the seminal work of Mandelstam and Tamm~\cite{MandelstamJPhys1945,mandelstam1991uncertainty}, it is known that the quantum speed is tightly bound by the Schr\"{o}dinger-Robertson (SR) uncertainty relation~\cite{RobertsonPR1929,schrodinger1999heisenberg}. This gives the QSLs a fundamental connotation that links the minimal time to attain a quantum state transformation to the energy dispersion imposed by the system Hamiltonian. Obviously, correlation functions constructed from the statistics of measurement outcomes, recorded by measuring $A$ and $B$, have to respect similar time constraints.

In the present analysis, we focus on situations in which the initial quantum state does not commute with the measurement observables. In this non-commutative case, correlation functions are the \textit{Kirkwood-Dirac quasiprobabilities} (KDQ)~\cite{yunger2018quasiprobability,ArvidssonShukurJPA2021,DeBievrePRL2021,LostaglioKirkwood2022,SantiniPRB2023,BudiyonoPRAquantifying,wagner2023quantum,ArvidssonShukur2024review}, which in general are complex numbers whose real part can be also negative. The presence of negative real parts and of imaginary parts different from zero does not allow for building up a single joint distribution of events at multiple times according to classical probability theory~\cite{johansen2007quantum,LostaglioKirkwood2022,GherardiniTutorial,ArvidssonShukur2024review}. In particular, negativity has been interpreted as a meaningful trait of \textit{non-classicality}~\cite{SpekkensPRL2008}, given that negative quasiprobabilities, as well as anomalous weak values, are a witness of quantum contextuality~\cite{HofmannPRL2012,PuseyPRL2014,Dressel2014colloquium,Hofer2017quasi,KunjwalPRA2019,SchmidPRA2024,SchmidQuantum2024}, as experimentally confirmed in~\cite{PiacentiniPRL2016,piacentini2016measuring,CiminiQST2020}. Beyond its fundamental relevance, negativity has recently found experimental application to enhance both phase estimation~\cite{lupu2021negative} and work extraction~\cite{hernandez2022experimental}, beyond classical limits. In addition, Ref.~\cite{hernandezArXiv2024Interfero} introduces an interferometric procedure to reconstruct the imaginary parts of two-time KDQ distributions, thus giving access to non-commutativity.

In this work, we address the following questions: ``Can we predict the time at which a two-time correlator, in terms of KDQ, can become non-positive? And, can we take an advantage of such a prediction in a quantum technology application?'' As main results, we demonstrate that the QSL predicts the \textit{minimal time} for the emergence of non-positive KDQ, while measuring two non-commuting observables at distinct times. This goes beyond the QSL for the average of a single observable~\cite{mohan2022quantum,ShrimaliPRA2024}, and finds application in the interdisciplinary field of evaluating the energetics of quantum computing gates~\cite{VeitchNJP2012,GardasSciRep2018,buffoni2020thermodynamics,CimininpjQI2020,Deffner2021EPL,BuffoniPRL2022,stevens2022energetics,Aifer2022NJP,Gianani2022diagnostics,śmierzchalski2023efficiency,Aifer2023PRXQ}, as well as quantum synchronization~\cite{Aifer2024PRL}. Moreover, the fact that the probability associated with a measurement-outcome pair is described by a non-positive KDQ outlines the role played by quantum coherence or correlations as a quantum resource. For instance, in non-equilibrium work processes our QSLs can be used to identify, and possibly reduce, the time corresponding to the largest enhancement of work extraction due to the negativity of the real part of some KDQ. Knowing such a time, which can be obtained without solving the system dynamics, allows to derive a tight bound on the work extraction power. We illustrate this for the two-qubit controlled-unitary gate of Ref.~\cite{CimininpjQI2020}. 

\section{Kirkwood-Dirac quasiprobabilities}

We start with notions and notations. Let $\rho$ be a density operator, and $A$ and $B$ two distinct quantum observables (Hermitian operators), evaluated at times $t=0$ and $t>0$ respectively. The two observables have the spectral decompositions $A=\sum_{\ell}a_{\ell}A_{\ell}$ and $B=\sum_{j}b_{j}B_{j}$ with $A_{\ell}=A_{\ell}^2$ and $B_{j}=B_{j}^2$. In the present analysis, the evolution of the quantum system is described by a unitary operator $U \equiv {\rm exp}(-i Ht / \hbar)$, where $H$ denotes the system Hamiltonian, $\hbar$ is the reduced Planck constant and $t$ the evolution time. In the following, we will also take into account the evolution of the quantum observable $B$, which is given by $B(t) \equiv U^{\dagger}B U$ in Heisenberg picture.

In general, KDQ are complex numbers and constitute a family of distributions of quasiprobabilities~\cite{ArvidssonShukur2024review,GherardiniTutorial}. For the purpose of this work, throughout the whole article, we take a representation of KDQ defined by two-time quantum correlators. In this regard, we can describe the statistics of the measurement-outcome pairs $(a_{\ell},b_{j})$, which occurs from evaluating the two quantum observables $A$ and $B$ at times $0,t$, with the KDQ
\begin{equation}\label{eq:def_quasiprob}
q_{\ell,j}(t) \equiv \tr{ \rho \, A_{\ell} \, B_j(t)},
\end{equation}
where $B_j(t) = U^\dagger B_j U$ is the projector $B_j$ evolved in the Heisenberg picture. In agreement with the no-go theorems of Refs.~\cite{PerarnauLlobetPRL2017,LostaglioKirkwood2022}, if the commutator of $\rho$ and $A$ is equal to zero ($[\rho,A]=0$), then the KDQ is a non-negative real number---as a standard probability obeying Kolmogorov axioms---and is given by the two-point measurement (TPM) scheme~\cite{campisi2011colloquium}.

The statistics provided by the TPM scheme can be experimentally assessed by a procedure based on sequential measurements, as in the classical case. On the contrary, as surveyed in~\cite{LostaglioKirkwood2022,GherardiniTutorial,ArvidssonShukur2024review}, the KDQ can be obtained via a reconstruction protocol that is able to preserve information on the non-commutativity of $\rho$ and the measurement observables. 
The most relevant properties of KDQ are: (i) $\sum_{\ell,j}q_{\ell, j}(t)=1$ $\forall t$; (ii) the \textit{unperturbed} marginals are recovered: $\sum_{\ell}q_{\ell, j}(t)=p_{j}(t)=\tr{\rho B_j(t)}$ and $\sum_{j}q_{\ell, j}(t)=p_{\ell}(0)=\tr{\rho A_{\ell}}$. The unperturbed marginal at time $t$ cannot be obtained by the TPM scheme if $[\rho, A] \neq 0$; (iii) Linearity in the initial state $\rho$; (iv) KDQ are equal to the joint probabilities 
\begin{equation}
\label{eq:TPM_joint}
p_{\ell, j}^\mathrm{TPM}(t) \equiv \tr{A_{\ell} \, \rho \, A_{\ell} \, B_j(t)} 
\end{equation}
determined by the TPM scheme when $[\rho, A] = 0$ $\forall \rho$. In the following, as customary in the literature, we will denote the real part of the KDQ as Margenau-Hill quasiprobability (MHQ)~\cite{HalliwellPRA2016,levy2020quasiprobability,hernandez2022experimental,PeiPRE2023}.

\subsection{Criteria of non-classicality}
\label{sec:criteria_of_nonclassicality}

As mentioned in the introduction, the KDQ may be non-classical. To quantify how far a KDQ is from being a positive real number, any of the following ``quantumness'' criteria can be used:
\begin{enumerate}
    \item[(i)] $\Re{q_{\ell, j}(t)} < 0$,
    \item[(ii)] $\left\lvert\Im{q_{\ell, j}(t)}\right\rvert > s_\mathrm{th}$,
\end{enumerate}
where $s_\mathrm{th} > 0$ is a threshold value beyond which we can say that the imaginary part is significantly different from zero. This value depends on the specific system under consideration, and we will provide an example in Section~\ref{sec:two-qubit-example}. We may similarly define a threshold for the real part without much modification to our results, although we do not follow this approach to keep the notation light. One of our goals in the following is to find how {\rm fast} the criteria (i) and (ii) can be satisfied~\footnote{A minimal time to non-positivity can be attained even for the quantumness criterion $\lvert\Re{q_{\ell, j}(t)} - p_{\ell, j}^\mathrm{TPM}(t)\rvert > k_\mathrm{th}$, with $k_\mathrm{th}$ a given threshold value, as discussed in~\cite{HePRA2024}.}.

\subsection{Real and imaginary parts of KDQ as expectation values of observables}

Our strategy is to bound the rate of change of the real and imaginary parts of the KDQ. We can write these two quantities as the expectation values of two different observables. In fact, defining the operators
\begin{equation}
    \rho_\ell
    \equiv
    \frac{\{\rho, A_\ell\}}{2}
    \quad\text{and}\quad
    \sigma_\ell
    \equiv
    \frac{[\rho, A_\ell]}{2i},
\end{equation}
which are time-independent and Hermitian, we can write
\begin{align}
    \Re{ q_{\ell, j}(t)}
    & = \tr{\rho_\ell \,B_j(t)}
    \equiv
    \expval{\rho_\ell}_{j,t},
    \label{eq:kdq_real}
    \\
    \Im{ q_{\ell, j}(t)}
    & = \tr{\sigma_\ell \,B_j(t)}
    \equiv
    \expval{\sigma_\ell}_{j,t}. 
\end{align}
Notice that, if $\sigma_\ell$ is the null matrix, then $\Im{ q_{\ell, j}(t)}=0$ for any t and $\Re{ q_{\ell, j}(t)}$ reduces to $p_{\ell, j}^\mathrm{TPM}(t)$.
If $[\rho,A_{\ell}]\neq 0$ for a given initial density operator $\rho$, then the KDQ of Eq.~\eqref{eq:def_quasiprob} can loose positivity.
This occurs if $B_j(t)$ in Eq.~\eqref{eq:def_quasiprob} is a projector onto the negative eigenspace of $\rho_\ell$~\cite{LostaglioKirkwood2022}.

We are going to bound the evolution of quantities of the type
\begin{equation}
\langle X \rangle_{j,t} \equiv \tr{X B_j(t)}.
\end{equation}
Up to normalization, $B_j(t)$ is in fact a density operator, which justifies the notation of expectation value, albeit it evolves in the Heisenberg representation according to the Hamiltonian $-H$. In the next section, we provide bounds for the evolution of the expectation value of general observables.
For the remainder of this article, and for simplicity of notation, we assume that $B_j$ has unit trace. However, the bounds we will derive in the next sections are applicable in the more general case where the trace of $B_j$ is not equal to $1$, by re-scaling $B_j$ as $B_j = \tr{B_j} \tilde B_j$ where the operator $\tilde B_j$ has unit trace. We can then use the linearity of the trace in $\expval{X}_{j,t}$ to bring $\tr{B_j}$ outside of the trace and make use of the same bounds.

\section{QSL for expectation values of observables}

We first derive a result that concerns the time-derivative of the expectation value $\langle X \rangle_{t} \equiv \tr{ X\rho(t)}$ with respect to an arbitrary density operator $\rho$ at a given time $t$.

\subsection{Bounding the rate of change using the SR uncertainty relation}

Before we present our original contributions, in the next sub-sections, we begin with some well-known results based on the SR uncertainty relation~\cite{PhysRevX.12.011038}. The SR uncertainty principle for any given observables $X,Y$ and
density operator $\rho(t)$ states that 
\begin{equation}\label{eq:sr_uncertainty}
    \Delta X_{t}^2
    \Delta Y_{t}^2
    \geq
    \left[
        \expval{\frac{\{X, Y\}}{2}}_{t}
        -\expval{X}_{t}\expval{Y}_{t}
    \right]^2
    + \expval{\frac{[X, Y]}{2i}}_{t}^2,
\end{equation}
where $\{\cdot,\cdot\}$ is the anti-commutator, and $\Delta C_{t} \equiv \sqrt{\langle C^2 \rangle_{t} -\langle C\rangle_{t}^2}$ with $C=X,Y$.

The von Neumman equation for the density operator, evolving according to a Hamiltonian $H$ (assumed as time-independent), is written as $\dot{\rho}(t)=[H, \rho(t)] / (i\hbar) \equiv \{ \rho(t), L(t)\} / 2$ where $L$ is the symmetric logarithmic derivative. Differentiating $\langle X\rangle_{t}$ over time leads to
\begin{equation}\label{eq:uncertainty_bound}
    \left|
        \frac{\dd}{\dd t}
        \langle X \rangle_{t}
    \right|
    = 
    \left\lvert
        \left\langle
        \frac{\{X, L(t)\}}{2}
        \right\rangle_{t}
    \right\rvert
     \leq
     \Delta L_{t} \,
     \Delta X_{t},
\end{equation}
where we used the SR uncertainty relation~\eqref{eq:sr_uncertainty} (see also \cite{nicholson2020time,PhysRevX.12.011038}) and the fact that $\expval{L(t)}_{t} = 0$ for any $t$, meaning that $\Delta L_{t}=\sqrt{\langle L^2\rangle_t}$. By definition, $\Delta L_t$ generally represents the square root of the {\it quantum Fisher information} $\mathcal{F}_Q(t)$ computed with respect to $\rho(t)$: $\Delta L_t = \sqrt{ \mathcal{F}_Q(t) }$. It generally holds that $\Delta L_t \leq 2 \Delta H_t / \hbar$ $\forall t$, where the equality applies in the case where $\rho$ is a pure state, whereby we recover the QSL in Ref.~\cite{mohan2022quantum}. Moreover, as shown in \ref{app:proof2}, $\Delta L$ is time-independent in our case-study, provided $H$ is constant.

\subsection{Explicit bounds on the evolution of $\langle X \rangle_t$}

We now present a novel contribution of this work. We derive explicit time-dependent bounds on $\expval{X}_t$ by bounding $\Delta X_{t}$ from above, following \ref{app:proof3}.
From Eq.~\eqref{eq:uncertainty_bound}, we then arrive at the differential inequality
\begin{equation}\label{eq:diff_inequality}
\left\lvert
    \frac{\dd}{\dd t}
    \langle X \rangle_{t}
\right\rvert
\leq
\Delta L \,
\sqrt{
    (x_1 \!+ x_d) \langle X \rangle_{t}
    - \langle X \rangle^2_{t}
    - x_1 x_d
}\,,
\end{equation}
where $x_1 \leq \cdots \leq x_d$ are the eigenvalues of $X$. As a result, integrating \eqref{eq:diff_inequality} results in the bounds (see \ref{app:proof4})
\begin{eqnarray}
    \langle X \rangle_{t} &\geq&
    \bound{X}{\tau_0 + \Delta L \, t},\label{eq:expval_lower_bound} \\
    \langle X \rangle_{t} &\leq&
    \bound{X}{\tau_0 - \Delta L \, t},\label{eq:expval_lower_bound_im}
\end{eqnarray}
where the function $\bound{X}{\tau}$, which we define below and is illustrated in Fig.~\ref{fig:Efunction}, interpolates the maximum and minimum eigenvalues of $X$:
\begin{equation}\label{eq:formal_lower_bound}
    \bound{X}{\tau} \equiv
    \begin{cases}
         x_d \,, &
         \tau \leq 0 \\
         x_d
         \cos^{2}(\frac{\tau}{2}) + x_1 \sin^{2}(\frac{\tau}{2})
         \,, &
         0\leq \tau \leq \pi \\
         x_1, &
         \tau \geq \pi \,.   
    \end{cases}
\end{equation}
In Eqs.~\eqref{eq:expval_lower_bound}-\eqref{eq:expval_lower_bound_im}, the initial angle $\tau_0$ is implicitly defined by the equality $\bound{X}{\tau_0} = \langle X \rangle_{t=0}$.
We can write it explicitly as $\tau_0 = \tau(X, \langle X\rangle_{0})$ with
\begin{equation}
    \tau(X, x)
    \equiv
    \arccos\left({
    \frac{2x - x_1 - x_d}{x_1 - x_d}
    }\right).
    \label{eq:interpolation_angle}
\end{equation}
The function $\tau(X, x)$ represents the ``interpolation angle'' quantifying where the value $x$ falls between $x_1$ and $x_d$ (minimal and maximal eigenvalues of $X$). This notation will be useful in the next section.

\begin{figure}[t]
    \centering
    \includegraphics[width=0.7\columnwidth]
    {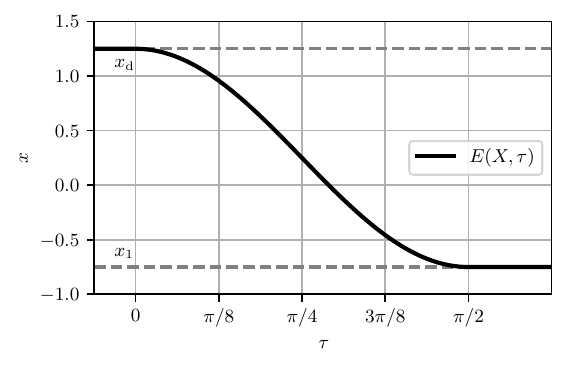}
    \caption{We illustrate the function $E(X, \tau)$, defined in Eq.~\eqref{eq:formal_lower_bound}, that enters the lower and upper bounds~\eqref{eq:expval_lower_bound}-\eqref{eq:expval_lower_bound_im}.
    It interpolates the minimum ($x_\mathrm{1}$) and maximum ($x_\mathrm{d}$) eigenvalues of $X$, using a trigonometric function of the angle $\tau$. The inverse function $\tau(X, x)$, defined in Eq.~\eqref{eq:interpolation_angle}, returns the angle $\tau_x$ such that $E(X, \tau_x)=x$.}
    \label{fig:Efunction}
\end{figure}

\subsection{Bounds on derivatives and unified bounds}

We can derive a further bound by noting that the derivative $\dd\langle X \rangle_t/ \dd t = \langle \dot X \rangle_t$ is itself the expectation value of the observable $\dot X \equiv [X, H] / i\hbar$. Therefore, $\langle \dot X \rangle_{t}$ can be bounded by \eqref{eq:expval_lower_bound} so that we obtain closed-form bounds as
\begin{equation}\label{eq:kdq_real_analytic_bound_derivative}
    \langle X \rangle_{t}
    \geq
    \expval{X}_{0} +
    \int_{0}^{\Delta L t} 
    \bound{\dot X}{\tau_0' + \xi} \, \dd \xi,
\end{equation}
where $\tau_0'$ is implicitly defined by the equality $\bound{\dot X}{\tau_0'}=\langle \dot X \rangle_{0}$ for $t=0$. The \textit{unified} lower-bound on $\langle X \rangle_{t}$ is thus the maximum between the right-hand-sides of Eqs.~\eqref{eq:expval_lower_bound} and~\eqref{eq:kdq_real_analytic_bound_derivative}.

\subsection{Saturation of the bounds} 

Using the Hamiltonian
$H = \frac{\hbar\omega}{2}\left(\ketbra{x_1}{x_d} + \ketbra{x_d}{x_1}\right),$
for some real number $\omega$,
and $\rho = \ketbra{x_d}{x_d}$ at $t=0$, the evolution of $\expval{X}_{t}$ exactly matches the right-hand side of the bound~\eqref{eq:expval_lower_bound}. The reader can find the technical details in \ref{app:proof5}.

This construction is important as it highlights the geometric interpretation of the inequality \eqref{eq:expval_lower_bound}, whose right-hand side represents the \textit{fastest path}, steered by $H$, to go from the maximum to the minimum eigenvalue of $X$.

\section{QSL for Kirkwood-Dirac quasiprobabilities}

Following the results of the previous section, we determine time-dependent bounds on the time-derivative of KDQ $q_{\ell, j}(t)$.

The analysis from Eq.~\eqref{eq:uncertainty_bound} to Eq.~\eqref{eq:kdq_real_analytic_bound_derivative} is valid for a generic Hermitian observable $X$. We now specialize it to the real and imaginary parts of KDQ by setting $X = \rho_{\ell}, \sigma_{\ell}$ in Eqs.~\eqref{eq:expval_lower_bound}-\eqref{eq:expval_lower_bound_im} respectively. This leads us to the bounds
\begin{align}
\Re{q_{\ell,j}(t)}
& \geq
\bound{\rho_\ell}{\tau^\mathrm{re}_{\ell,0} + \Delta L_j \, t},
\label{eq:kdq_real_analytic_bound} \\
\Im{q_{\ell, j}(t)}
& \geq
\bound{\sigma_\ell}{\tau^\mathrm{im}_{\ell,0} + \Delta L_j \, t},
\label{eq:kdq_imag_analytic_bound}
\end{align}
where $\tau^\mathrm{re}_{\ell,0}, \tau^\mathrm{im}_{\ell,0}$ are defined by the equality of the bounds at $t=0$, and can be computed using Eq.~\eqref{eq:interpolation_angle}. Moreover, $L_j(t)$ is the symmetric logarithmic derivative built on the Hamiltonian $-H$ and the state $B_j(t)$ such that $dB_{j}(t)/dt = [ -H, B_{j}(t) ]/i\hbar = \{ B_{j}(t), L_j(t) \}/2$. Thus, the standard deviation $\Delta L_{j,t}$ is also computed with respect to $B_j(t)$, i.e., $\Delta L_{j,t} = \sqrt{ \langle L_j(t)^2\rangle_t } = \tr{ L_j(t)^{2} B_j(t) }$. It is worth noting that, following the same steps in \ref{app:proof2}, we can prove that $\Delta L_j$ is time-independent for all $j$, provided $H$ is constant.

Finally, replacing $\dot X = \dot \rho_\ell, \dot \sigma_\ell$ in Eq.~\eqref{eq:kdq_real_analytic_bound_derivative} can lead to more refined bounds on the real and imaginary parts of the KDQ, as also shown in Fig.~\ref{fig:kdq_real_bounds} for the internal energy variations of a qubit subject to a controlled-unitary gate.

\subsection{Commutative limit} 

Let us focus on the case of $[\rho, A_\ell] = 0$, whereby the KDQ are equal to the joint probabilities returned by the TPM scheme [Eq.~\eqref{eq:TPM_joint}] that, by definition, are always positive and $\in[0,1]$. If $[\rho, A_\ell] = 0$, then the imaginary part of the KDQ is zero. Given $A_{\ell} \neq \mathbb{I}$ with $\mathbb{I}$ the identity operator, the minimum eigenvalue of $\rho_{\ell}$ (here equal to $A_{\ell}\rho A_{\ell}$) is $r_1 = 0$.
As a result, from Eq.~\eqref{eq:kdq_real_analytic_bound}, one has that
\begin{equation}\label{eq:bound_on_p_TPM}
    p^\mathrm{TPM}_{\ell,j}(t)
    \geq
    r_{d}\cos^{2}\left(
     \frac{\tau_{\ell,0}^\mathrm{re} +\Delta L_j  \, t}{2}
    \right).
\end{equation}
Hence, as expected, no negativity can be observed from Eq.~\eqref{eq:bound_on_p_TPM}. 

\subsection{Time to non-positivity}

The bounds Eqs.~\eqref{eq:kdq_real_analytic_bound}-\eqref{eq:kdq_imag_analytic_bound} allow us to derive minimal times to non-positivity.
Following the two criteria (i)-(ii) of non-classicality in Section~\ref{sec:criteria_of_nonclassicality}, we ask how long it takes for the right-hand sides of \eqref{eq:kdq_real_analytic_bound}-\eqref{eq:kdq_imag_analytic_bound} to go from their initial values to the target values of 0 and $s_\mathrm{th}$ respectively, namely what are the times $T_{\ell, j}^\mathrm{re},T_{\ell, j}^\mathrm{im}$ such that $E(\rho_\ell, \tau_{\ell,0}^\mathrm{re} + \Delta L_j T_{\ell,j}^\mathrm{re})=0$ and $E(\sigma_\ell, \tau_{\ell,0}^\mathrm{im} + \Delta L_j T_{\ell,j}^\mathrm{im})=s_\mathrm{th}$.
Using Eq.~\eqref{eq:interpolation_angle}, the initial and target values can be expressed in terms of the function $\tau(X, x)$. As a result, for the $(\ell,j)$th KDQ, the minimal times $T_{\ell, j}^\mathrm{re},T_{\ell, j}^\mathrm{im}$ read respectively as
\begin{align}
    T_{\ell, j}^\mathrm{re}
    &=
    \frac{
        \tau(\rho_\ell, 0)
        -\tau^\mathrm{re}_{\ell,0}
    }{\Delta L_j}\,,
    \label{eq:time_qsl}
    \\
    T_{\ell, j}^\mathrm{im}
    &=
    \frac{
        \tau(\sigma_\ell, s_\mathrm{th})
        -\tau^\mathrm{im}_{\ell,0}
    }{\Delta L_j}\,.
    \label{eq:time_qsl_im}
\end{align}
Equation~\eqref{eq:time_qsl} coincides with the Mandelstam's and Tamm's result~\cite{MandelstamJPhys1945,mandelstam1991uncertainty} in the following circumstance: $\rho = \ketbra{\psi}{\psi}$ is a pure state, $A_\ell$ is the identity (i.e., no measurement is performed at time $t=0$), and $B_j(t)=\rho$ at a given time $t$. Under these assumptions, indeed, $T_{\ell j}^\mathrm{re} = \hbar \pi \, / \, ( 2\Delta H )$. Hence, in analogy, we can still interpret $T_{\ell j}^\mathrm{re}$ as the minimum time for a quantum system to evolve towards an orthogonal state, even in the more general setting with non-commutative observables and initial state.

\section{Enhancing power extraction}

Non-commutativity between the initial state $\rho$ and a time-dependent Hamiltonian implementing a work protocol can be a resource to \textit{enhance quantum work extraction} beyond what can be achieved by any classical system~\cite{hernandez2022experimental,Touil2022JPA,GherardiniTutorial}. To this end, let us recall the definition of the extractable work $W_\mathrm{ext}(t)$ at time $t$ in a coherently-driven closed quantum system. We make use of the spectral decomposition of the system Hamiltonian, $H(t_k)=\sum_{n}E_{n}(t_k)\Pi_{n}(t_k)$, with $k=1,2$, $n\in\{\ell,j\}$ and $\{\Pi_{n}(t_k)\}$ denoting the sets of projectors over the energy basis. $W_\mathrm{ext}(t)$ is thus defined as~\cite{pusz1978passive,Alicki1979}
\begin{eqnarray}
     W_\mathrm{ext}(t) &\equiv& -\langle w(t)\rangle = \sum_{\ell, j} \left( E_{\ell}(0) - E_{j}(t) \right) q_{\ell, j}(t) =\nonumber \\
     &=& \tr{\rho \, H(0)} - \tr{U\rho \, U^{\dagger}H(t)},
\end{eqnarray}
where $w(t) \equiv E(t) - E(0)$ is the \textit{stochastic work} so that $w_{\ell,j}(t) \equiv E_{j}(t) - E_{\ell}(0)$.
A necessary condition to enhance $W_\mathrm{ext}(t)$, beyond the maximum value achievable classically, is that $\Re{ q_{\ell,j}(t)} < 0$ for some time $t$~\cite{hernandez2022experimental}.
Negative MHQ allow for \textit{anomalous energy transitions}, i.e., work realizations $w_{\ell,j}$ that occur with a negative quasiprobability.
Thus, if negative MHQ are associated with positive $w_{\ell,j}$, then the extractable work is boosted.

Considering time constraints is important in quantum engines and energy conversion devices~\cite{UzdinPRX2015,MyersAVS2022,CangemiArXiv2023,ReviewQBattery}, where the energy power depends on the times in which the strokes of a given machine are accurately performed. Thus, for work extraction purposes, one would like to achieve the maximum possible value of $W_\mathrm{ext}(t)$ in the shortest possible time. This means that, if also a boost of work extraction due to negative MHQ is included, then one needs to derive the minimum time at which $\Re{q_{\ell,j}(t)}<0$ for positive work realizations $w_{\ell, j}(t) > 0$. At the same time, the MHQ associated with negative $w_{\ell, j}(t)$ have to be positive. Hence, the optimization of the work extraction in finite-time transformations requires to maximize the \textit{work extraction power}, which is defined by $\mathcal{P}(t) \equiv - \langle w(t)\rangle / T_\mathrm{max}$, where $T_\mathrm{max}$ is the time for maximum extractable work~\cite{Touil2022JPA}.

The optimal value of $\mathcal{P}$ comes from a trade-off between maximum extractable work (even enhanced by negativity) and minimum time period $T_\mathrm{max}$. Such an optimal $\mathcal{P}$ could be bounded by an approximated power function computed using the minimal time $T_{\ell j}^\mathrm{re}$ to negative quasiprobabilities. 

\subsection{Two-qubit example}
\label{sec:two-qubit-example}

We conclude the analysis by verifying our bounds on both the real and imaginary parts of KDQ, as well as the bound on $\mathcal{P}$ using $T_{\ell j}^\mathrm{re}$ for the two-qubit controlled-unitary gate experimentally realized in~\cite{CimininpjQI2020}. The gate evolves according to the time-independent Hamiltonian 
\begin{equation}
H_{\rm 2qubits} = \frac{\omega_L}{2} (Z_1 + Z_2) + \frac{\omega_\mathrm{int}}{2}{ \ketbra{1}{1} \otimes X},
\end{equation}
where $Z_i$ is the $Z$-Pauli matrix applied to the qubit $i$, $X$ is the $X$-Pauli matrix, and $|0\rangle,|1\rangle$ are respectively the ground and excited states of the local Hamiltonian of each qubit. The first qubit acts as a `control' knob: if it is in the excited state, the second qubit (`target') undergoes a rotation of a parameterized angle.

In this process, the internal energy of the target qubit changes with time, as provided by computing the partial trace of the two-qubit state with respect to the degrees of freedom of the control qubit. Assuming the control qubit can be manipulated at will, we interpret the internal energy variation of the target qubit as thermodynamic work exerted by the control qubit. We compute the KDQ of the internal energy variation of the target qubit by setting $A_\ell = \ketbra{\ell}{\ell}$ and $B_j = \ketbra{j}{j}$ with $\ell,j=0,1$. Preparing the global system in the state $|-\rangle |1 \rangle$, where $|-\rangle \equiv (|0\rangle - |1\rangle) / \sqrt{2}$, leads to non-positivity of the computed KDQ. Here, negativity is a signature of extractable work from the target qubit, in a regime where the corresponding value returned by the TPM scheme, $W_\mathrm{ext}^\mathrm{TPM}(t) \equiv \sum_{\ell,j}( E_{\ell}(0) - E_{j}(t)) p_{\ell, j}^\mathrm{TPM}(t)$, is zero for any parameters choice, due to the state-collapse upon the first energy measurement of the TPM scheme and the specific choice of the initial state for the global system.

\begin{figure*}[t]
    \centering
    \includegraphics[width=\columnwidth]{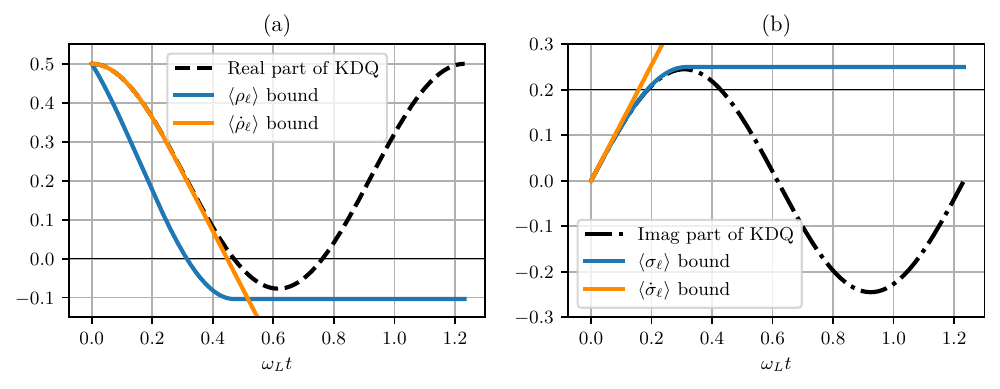}
    \caption{
    Quasiprobability associated to the internal energy variations of the target qubit in a two-qubit controlled-unitary gate with Hamiltonian $H_{\rm 2qubits}$. \textit{Black line}: Real [panel (a)] and imaginary [panel (b)] parts of the KDQ $q_{1,1}(t)$ using 
    $A_1 = B_1 = \ketbra{1}{1}$, with $\{|0\rangle,|1\rangle\}$ computational basis. Globally the quantum gate is initialized in the pure state $\ket{-}\ket{1}$. Negativity in $\Re{q_{1,1}(t)}$ is observed choosing $\omega_L=1$ and $\omega_\mathrm{int} = 5$. \textit{Blue and orange lines}: Lower-bounds to the quasiprobability, both derived from the uncertainty principle \eqref{eq:sr_uncertainty}. They are respectively obtained by constraining the evolution of $\expval{\rho_1}_{1,t}, \expval{\sigma_1}_{1,t}$ [Eqs.~\eqref{eq:kdq_real_analytic_bound}-\eqref{eq:kdq_imag_analytic_bound}] (blue lines), and $\expval{\dot \rho_1}_{1,t}, \expval{\dot\sigma_1}_{1,t}$ [Eq.~\eqref{eq:kdq_real_analytic_bound_derivative} with $X=\dot\rho_\ell, \dot\sigma_\ell$] (orange lines).
    }
    \label{fig:kdq_real_bounds}
\end{figure*}

\begin{figure}[t]
    \centering
    \includegraphics[width=0.6\columnwidth]{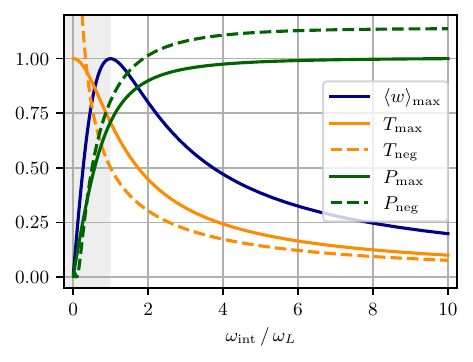}
    \caption{The time to negativity $T_\mathrm{neg}$ (dashed orange) acts as a lower-bound to the time $T_\mathrm{max}$ (solid orange), at which the maximum average work $\langle w\rangle_\mathrm{max}$ (blue) is extracted from the target qubit under conditions of Fig.~\ref{fig:kdq_real_bounds}.
    The power at maximum energy extraction, $\mathcal{P}_\mathrm{max} = \langle w\rangle_\mathrm{max} / T_\mathrm{max}$ (solid green), can be well-approximated by the power $\mathcal{P}_\mathrm{neg}$ computed at $T_\mathrm{neg}$ (dashed green). $\mathcal{P}_\mathrm{neg}$ is an upper-bound of $\mathcal{P}_\mathrm{max}$.
    This reasoning cannot be applied to the light gray region, since the negativity is not present. The value of the energies, times and powers are respectively normalized by $E_\mathrm{ref}=0.5$, $t_\mathrm{ref} = \pi$ and $\mathcal{P}_\mathrm{ref}\approx 0.32$.}
    \label{fig:power}
\end{figure}

In Fig.~\ref{fig:kdq_real_bounds}(a-b) we plot the real and imaginary parts of the quasiprobability $q_{1,1}(t)$ (i.e., $\ell,j=1$), respectively, together with the lower-bounds \eqref{eq:kdq_real_analytic_bound_derivative} with $\dot X = \dot \rho_\ell, \dot \sigma_\ell$, and \eqref{eq:kdq_real_analytic_bound}-\eqref{eq:kdq_imag_analytic_bound}, for $\omega_L=1$ and $\omega_\mathrm{int}=5$. One has to take the tightest between the blue and orange curves as the lower-bound to $\mathrm{Re}\{ q_{1,1}(t) \}, \mathrm{Im}\{ q_{1,1}(t) \}$ for any time $t$. In Fig.~\ref{fig:kdq_real_bounds}(a), the first time to negativity (given by crossing the threshold at $\mathrm{Re}\{ q_{1,1} \}=0$) is well identified by the lower-bound \eqref{eq:kdq_real_analytic_bound_derivative}, while the first time $\mathrm{Im}\{ q_{1,1}(t) \}$ touches the threshold set at $0.2$ is provided by the lower-bound \eqref{eq:kdq_imag_analytic_bound}. This choice to set the threshold of $\mathrm{Im}\{ q_{1,1} \}$ at $0.2$ is dictated by the evidence that, in possible experiments as recently in \cite{hernandezArXiv2024Interfero}, the error bars in reconstructing the imaginary part of a quasiprobability are around $10\%$ of the actual value.

In Fig.~\ref{fig:power}, we plot the time to negativity $T_\mathrm{neg}$, which is defined as the time at which the bound of $\langle\dot\rho_\ell\rangle_{j,t}$ [Eq.~\eqref{eq:kdq_real_analytic_bound_derivative} with $X=\dot\rho_\ell$] in Fig.~\ref{fig:kdq_real_bounds} reaches zero. $T_\mathrm{neg}$ is a lower-bound of the time $T_\mathrm{max}$ at which the maximum amount of energy can be possibly extracted from the target qubit. The work extraction power $\mathcal{P}_\mathrm{neg}$ obtained at $T_\mathrm{neg}$ turns out to be a good approximation for the power $\mathcal{P}_\mathrm{max}$ at $T_\mathrm{max}$. Notably, it is an upper-bound that can be experimentally measured via local measurements.

\section{Conclusions}

We have derived time-dependent bounds for the Kirkwood-Dirac quasiprobabilities, which stem from using the Schr\"{o}dinger-Robertson uncertainty relation on the time-derivative of such quasiprobabilities. Our derivation can be interpreted as an extension to non-commutative operators of the quantum speed limit bound obtained by Mandelstam and Tamm.

The first consequence of our results is to determine the minimal time it takes for the real part of a KDQ to become negative, and for the corresponding imaginary part to go from zero to exceed a threshold value. The minimal time to negativity has an application as a bound on the maximum power of a finite-time work extraction process. Consequently, we suggest to further investigate more complex quantum gates~\cite{FedorovReview2022}, and devices for energy conversion~\cite{UzdinPRX2015,MyersAVS2022,CangemiArXiv2023} including quantum batteries~\cite{ReviewQBattery}. 
Moreover, the quantum speed limit time-bound on KDQ could also predict abrupt or anomalous changes in out-of-time-ordered correlators (OTOC), which can be equivalently expressed as the characteristic function of a KDQ distribution~\cite{yunger2018quasiprobability,mohseninia2019strongly,GherardiniTutorial}. 

\section*{Acknowledgements}

S.S.P.\ acknowledges support from the ``la Caixa'' foundation through scholarship No.\ LCF/BQ/DR20/11790030, and from national funds by FCT -- Fundação para a Ciência e Tecnologia, I.P., through scholarship 2023.01162.BD, and in the framework of the projects UIDB/04564/2020 and UIDP/04564/2020, with DOI identifiers 10.54499/UIDB/04564/2020 and 10.54499/UIDP/04564/2020, respectively.
S.G.\ acknowledges financial support from the project PRIN 2022 Quantum Reservoir Computing (QuReCo), the PNRR MUR project PE0000023-NQSTI financed by the European Union--Next Generation EU, and the MISTI Global Seed Funds MIT-FVG Collaboration Grant ``Revealing and exploiting quantumness via quasiprobabilities: from quantum thermodynamics to quantum sensing''.
S.D.\ acknowledges support from the U.S. National Science Foundation under Grant No. DMR-2010127 and the John Templeton Foundation under Grant No. 62422.

\appendix

\section{Proof I: Bounds on the expectation value using the Hamiltonian and symmetric logarithmic derivative}\label{app:proof2}

In the main text, $B_j(t)$ plays the role of a (unnormalized) density operator $\rho$.
In this section, we derive bounds on the evolution-rate for the expectation value of an observable $X$, taken with respect to $\rho$.
Suppose $\rho$ evolves unitarily as $\rho(t) = U \rho(0) U^\dagger$, with $U = \exp(-i Ht / \hbar)$ for some constant Hermitian operator $H$, so that
\begin{equation}
    \frac{\dd\rho(t)}{\dd t}
    = \frac{[H, \rho(t)]}{i \hbar}.
    \label{eq:von_neumann_equation}
\end{equation}
Another way to write the evolution of $\rho$ is using the symmetric logarithmic derivative operator $L$ that is defined implicitly, for all times, by
\begin{equation}
    \frac{\dd \rho(t)}{\dd t}
    \equiv
    \frac{1}{2}
    \{\rho(t), L(t)\}.
    \label{eq:symlogdev_definition}
\end{equation}
The spectral decomposition $\rho = \sum_j p_j \ketbra{j}{j}$ allows us to write $L$ in the basis $\{ \ket{j} \}$, such that
\begin{equation}
    L_{ij}
    =
    \frac{2}{i \hbar}
    \frac{p_j - p_i}{p_i + p_j}
    H_{ij},
\end{equation}
where
$L_{ij}
\equiv
\langle i \vert L \vert j \rangle$
and
$H_{ij} \equiv \langle i \vert H \vert j \rangle$.
If $p_i = p_j = 0$, then we take $L_{ij}$ equal to zero.
Moreover, for a constant Hamiltonian $H$, we can take
\begin{equation}
L(t) = UL(0)U^\dagger,
\end{equation}
which is compatible with the definition of $L$ for any time $t$, that is,
\begin{equation}
\frac{\{\rho(t), L(t)\}}{2}
= U\frac{\{\rho(0), L(0)\}}{2}U^\dagger
= U\frac{[H, \rho(0)]}{i\hbar}U^\dagger
= \frac{[H, \rho(t)]}{i\hbar}
= \dot \rho(t).
\end{equation}

As a result, we have two ways to bound the evolution of the expectation value $\expval{X}_{t}$. The first way follows the Hamiltonian formalism:
\begin{equation}\label{eq:SM_1st_bound}
    \abs{\frac{\dd \expval{X}_{t}}{\dd t}}
    =
    \abs{\expval{
    \frac{[X, H]}{i \hbar}
    }_{t}}
    \leq
    \frac{2}{\hbar}\Delta H \, \Delta X_{t},
\end{equation}
where we used the SR uncertainty relation \eqref{eq:sr_uncertainty} with $Y = H$.

The second bound, making use of the symmetric logarithmic derivative, reads as
\begin{equation}\label{eq:SM_2nd_bound}
    \abs{\frac{\dd \expval{X}_{t}}{\dd t}}
    =
    \abs{\expval{
    \frac{\{X, L\}}{2}
    }_{t}}
    \leq
    \Delta L_{t} \, \Delta X_{t},
\end{equation}
where we used the replacement $Y=L(t)$. Note that, using the definition of $L$, we have $\expval{L(t)}_t = 0$ $\forall t$.

We note that, if $H$ is constant, then also $\Delta L$ is constant. Indeed, using that $L(t) = U L(0) U^\dagger$ like $\rho(t) = U \rho(0) U^\dagger$, we get:
\begin{equation}
    \Delta L_{t}^2
    =
    \expval{L(t)^2}_t
    =
    \tr{\rho(t) L(t)^2}
    =
    \tr{\rho(0) L(0)^2}
    =
    \Delta L^2 \quad \forall t,
\end{equation}
which implies that $\Delta L_{t}^2$ is constant. This fact allows to calculate the bounds at $t=0$ without solving for the system dynamics.

It turns out that, if $\textrm{rank}(\rho)=1$, then the two bounds \eqref{eq:SM_1st_bound} and \eqref{eq:SM_2nd_bound} are equal. This can be shown by setting $\rho = \ketbra{\varphi}{\varphi}$, whereby
\begin{equation}
    \frac{\Delta L^{2}}{2}
    =
    \tr{
        \left(
        \frac{\{\rho(t), L(t)\}}{2}
        \right)^2
    }
    =
    \tr{
        \left(
        \frac{[H, \rho(t)]}{i\hbar}
        \right)^2
    }
    = \frac{2 \Delta H^2}{\hbar^2},
\end{equation}
i.e., $\Delta L = 2 \, \Delta H / \hbar$, meaning that the two bounds \eqref{eq:SM_1st_bound} and \eqref{eq:SM_2nd_bound} coincide.

More in general instead ($\textrm{rank}(\rho) > 1$), the use of $\Delta L$ results in a tighter bound.
To see this, let us equate the right-hand sides of \eqref{eq:von_neumann_equation} and \eqref{eq:symlogdev_definition}, multiply them by $L(t)$, and take the trace. This results in
\begin{equation}
    \Delta L^2
    =
    \expval{\frac{[L(t), H]}{i\hbar}}_{t}
    \leq
    \frac{2}{\hbar}
    \Delta L \, \Delta H,
\end{equation}
where the inequality comes from the uncertainty principle. Therefore, in general,
\begin{equation}
    \Delta L \leq \frac{2}{\hbar} \, \Delta H
\end{equation}
that proves that using the symmetric logarithmic derivative provides a tighter bound.

\section{Proof II: Upper bound on the variance of a Hermitian operator}\label{app:proof3} 

When calculating the bounds on quasiprobabilities using the Schr\"{o}dinger-Robertson uncertainty relation, we are led to a differential equation depending on $\Delta X \equiv \sqrt{ \langle X^{2}\rangle - \langle X\rangle^{2} }$ with $X$ being a generic Hermitian operator. The quantity $\Delta X$ can be bounded using only the expectation value $\expval{X} = \tr{X \rho}$, namely
\begin{align}
    \Delta X^2
    &= \expval{X^2} - \expval{X}^{2}  \nonumber
    \\&\leq
    (x_1 + x_d) \expval{X} - x_1 x_d - \expval{X}^{2},
    \label{eq:bound_Delta_X}
\end{align}
where $x_1, x_d$ are respectively the lowest and highest eigenvalues of $X$.

To derive Eq.~\eqref{eq:bound_Delta_X}, we use the following proposition.
\begin{prop}\label{prop_1}
Let $X$ be a Hermitian operator with eigenvalues $x_1 \leq \cdots \leq x_d$ and $\rho$ a density operator. Then,
\begin{equation}\label{eq:X2_expval_upper_bound}
    \langle X^2 \rangle
    \leq
    (x_1 + x_d) \langle X\rangle
    - x_1 x_d.
\end{equation}
\end{prop}
\begin{proof}
    Notice that $x_1 \mathbb{I} \preceq X \preceq x_d \mathbb{I}$, where $\mathbb{I}$ is the identity operator. Hence, $(X - x_1\mathbb{I})$ and $(x_{d}\mathbb{I} - X)$ are both positive semidefinite operators, such that by multiplying them one has
    \begin{equation}\label{eq:proof_eq1}
        (X - x_1 \mathbb{I})(x_d \mathbb{I} -X) \succeq 0.
    \end{equation}
    Expanding and rearranging \eqref{eq:proof_eq1}, we get
    \begin{equation}
        X^2 \preceq (x_1 + x_d)X - x_1x_d \mathbb{I}.
    \end{equation}
    Therefore, taking the expectation value with respect to $\rho$, we arrive at the desired result.
\end{proof}
To conclude, in Figure~\ref{fig:X2_expval_upper_bound} we illustrate the inequality~\eqref{eq:X2_expval_upper_bound} and how it works.
\begin{figure}[ht!]
     \centering
     \includegraphics[width=0.65\columnwidth]{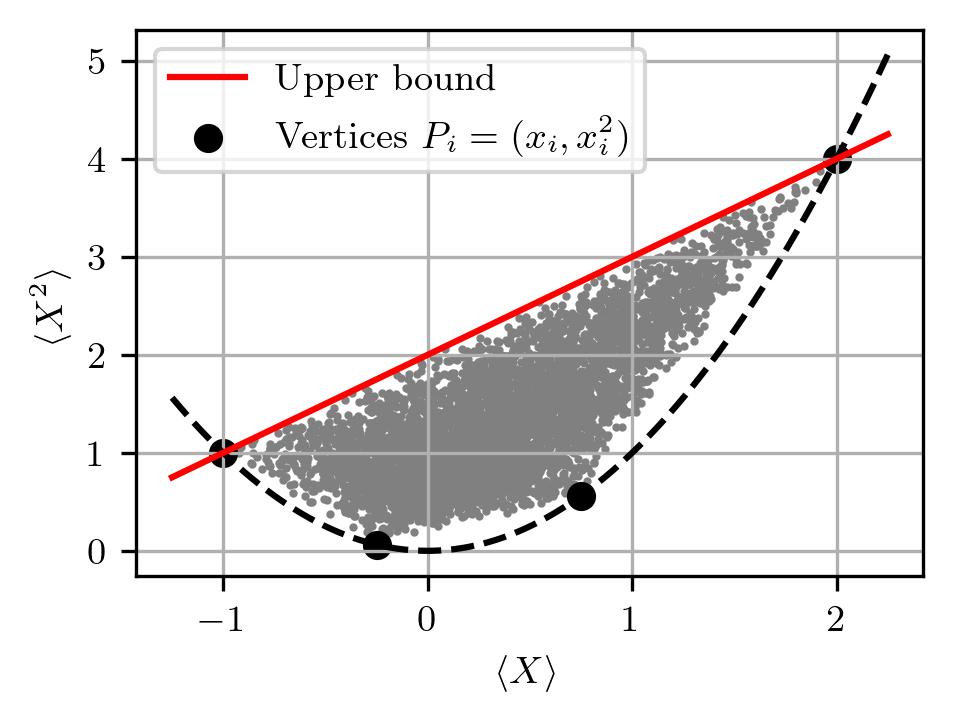}
     \caption{
     Illustration of inequality~\eqref{eq:X2_expval_upper_bound}. For a Hermitian operator $X$ with eigenvalues $x_1 \leq \cdots \leq x_d$, the set of possible points $(\langle X\rangle, \langle X^{2}\rangle)$ (gray dots) forms a convex polygon, whose vertices $P_i=(x_i, x_i^2)$ lie on a common parabolic path (dashed line). Since a parabola is a convex function, any $\langle X^2 \rangle$ can be bounded from above using the secant joining the vertices $P_1$ and $P_d$ (red line).
     }
     \label{fig:X2_expval_upper_bound}
\end{figure}

\section{Proof III: Solution of the differential inequality of Eq.~(\ref{eq:diff_inequality})}\label{app:proof4}

In this Appendix we solve the following differential inequality of Eq.~(\ref{eq:diff_inequality}) that we express in a more compact form as
\begin{equation}\label{eq:diffeq}
    \abs{\frac{dx}{dt}}
    \leq \omega \sqrt{ (a+b)x - x^2 - ab },
\end{equation}
where $x=x(t)$ is a time-dependent variable, and $\omega \geq 0$ and $a \leq b$ are two real constants. Also notice that the function $x$ is bounded as $a \leq x(t) \leq b$ for all times $t$.
By defining $x \equiv \alpha z + \beta$, with $\alpha \equiv (b - a)/2$ and $\beta \equiv (a + b) / 2$ (i.e., $a = \beta - \alpha$ and $b = \alpha + \beta$), we get the reduced form
\begin{equation}\label{eq:diff_eq_z}
    \abs{\frac{dz}{dt}} \leq \omega \sqrt{1 - z^2}  \,.
\end{equation}
Let us focus on the negative branch of the differential inequality \eqref{eq:diff_eq_z}, since the positive branch has an analogous solution. The negative branch of \eqref{eq:diff_eq_z} can be solved by substituting 
\begin{equation}\label{app:second_substitution}
z = \cos(\tau) 
\quad \Longleftrightarrow \quad
\dd z = -\sin(\tau) \, \dd\tau \,.   
\end{equation}
This is because the lower bound of \eqref{eq:diff_eq_z} entails the separable differential inequality  
\begin{equation}
    \frac{ \dd z }{ \sqrt{1-z^2} }
    \geq
    -\omega \, \dd t.
\end{equation}
Upon the substitution \eqref{app:second_substitution}, this simplifies to
\begin{equation}
    \frac{d\tau}{dt} \leq \omega
    \quad \Longrightarrow \quad
    \tau(t) \leq \tau_0 + \omega t \,,
\end{equation}
where $\tau_0$ is the value of $\tau$ at the initial time ($t=0$, in the main text of the paper) of the interval in which the differential inequality \eqref{eq:diffeq} is solved.
In this way, by replacing back $z$ and $x$, and using the fact that the cosine is decreasing in the region $[0, \pi]$, we get
\begin{equation}
    x(t)
    \geq
    \begin{dcases} 
        \alpha \cos(\tau_0 + \omega t) + \beta,
        & \text{for } t < t^* \\
        a,
        & \text{for }  t \geq t^*
        \label{eq:lower_bound}
    \end{dcases}
\end{equation}
where $\tau_0 = \arccos((x(0) - \beta) / \alpha) \in [0, \pi]$ and $t^{*}= (\pi - \tau_0) / \omega$. The solution is divided into branches because, despite the oscillating solution, the lower bound of $x(t)$ can never increase in value, given that the lower branch of the inequality \eqref{eq:diffeq} (symmetric with respect to $0$) is never positive. For the sake of clarity, see Fig.~\ref{fig:diffeq_solution} for a visual representation of the solution \eqref{eq:lower_bound}, concerning the real part of an actual Kirkwood-Dirac quasiprobability (KDQ). Therefore, there is a time interval (from 0 to $t^*$) for the analytical solution to be valid; beyond $t = t^*$, the lower bound must remain constant and equal to $a$.

Using standard trigonometric identities, we may rewrite the bound as
\begin{equation}
    x(t) \geq F(\tau_0 + \omega t),
\end{equation}
where 
\begin{equation}
    F(\tau)
    \equiv
    \begin{cases}
    b \cos^2(\frac{\tau}{2}) + a \sin^2(\frac{\tau}{2}) & 0 \leq \tau \leq \pi
    \\ 
    a & \tau \geq \pi
    \end{cases}.
\end{equation}

The positive branch of the differential equation \eqref{eq:diff_eq_z} can be integrated similarly to the negative branch, with the result that
\begin{equation}
    x(t) \leq F'(\tau_0 - \omega t),
\end{equation}
where $F'$ is similarly defined as
\begin{equation}
    F'(\tau)
    \equiv
    \begin{cases} 
    b \cos^2(\frac{\tau}{2}) + a \sin^2(\frac{\tau}{2}) & 0 \leq \tau \leq \pi
    \\ 
    b & \tau \leq 0 \,.
    \end{cases}
    \label{eq:upper_bound}
\end{equation}
Since we assumed only $t \geq 0$, the domain of the functions $F$ and $F'$ overlap only in the interval $[0, \pi]$ where they coincide. Therefore, we can write the bounds using a unified function $E$,
\begin{equation}
    E(\tau)\equiv
    \begin{cases}
    b & \tau \leq 0
    \\
    b \cos^2(\frac{\tau}{2}) + a \sin^2(\frac{\tau}{2}) & 0 \leq \tau \leq \pi
    \\ 
    a & \tau \geq \pi \,,
    \end{cases}
\end{equation}
and thus obtain
\begin{align}
    x(t) &\geq E(\tau_0 + \omega t), \label{eq:lower_bound_general} \\
    x(t) &\leq E(\tau_0 - \omega t), \label{eq:upper_bound_general} 
\end{align}
for $t \geq 0$.

\begin{figure}[ht!]
    \centering
    \includegraphics[width=0.65\columnwidth]{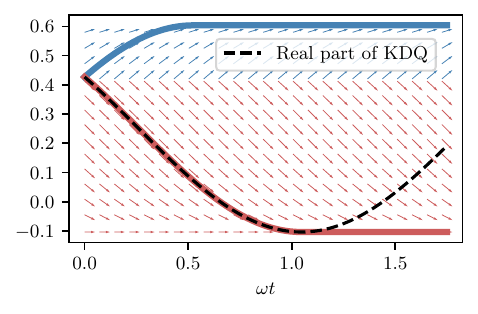}
    \caption{Upper and lower bounds (blue and red lines respectively) of the real part of the Kirkwood-Dirac quasiprobability $q_{1,1}(t)$ as a function of time (in dimensionless units). The bounds are obtained by integrating the differential inequality~\eqref{eq:diffeq}. It can be observed that the lower bound is saturated, being superimposed on the quasiprobability. The blue arrows represent the gradient field of the upper bound, while the red arrows correspond to the gradient field of the lower bound. The black dotted line depicts the time-evolution of $q_{1,1}(t)$. The quasiprobability drawn in the figure comes from taking $(\ket 0 + \ket 1) / \sqrt{2}$ as the initial state, and $A_1 = \ketbra{0}{0}$ as the projector of the initial measurement at $t=0$. The quantum system evolves unitarily under the Hamiltonian $H = \frac{\hbar\omega}{2} \, (\ketbra{r_1}{r_d} + \ketbra{r_d}{r_1})$, where $\ket{r_1}, \ket{r_d}$ are the eigenvectors of $\rho_1 \equiv \{\rho, A_1 \} / 2$ associated to the minimum and maximum eigenvalues $r_1$ and $r_d$, respectively. The projector of the second measurement (at time $t$) is $B_1 = \ketbra{\psi}{\psi}$, where $\ket{\psi} = \cos(\frac{\pi}{6}) \ket{r_1} - i \sin(\frac{\pi}{6}) \ket{r_d}$.}
    \label{fig:diffeq_solution}
\end{figure}

\section{Proof IV: Saturation of the bounds}\label{app:proof5}
The bounds~\eqref{eq:lower_bound_general} and \eqref{eq:upper_bound_general} can be saturated. In Fig.~\ref{fig:diffeq_solution}, we show an example where the real part of a KDQ is saturated. Let us now explain under which conditions the saturation of the bounds is achieved.

Let $X$ be an observable with eigenvalues $x_1 \leq \cdots \leq x_d$, corresponding to the eigenvectors $\ket{x_1}, \ldots, \ket{x_d},$ respectively.
Then, let us consider an initial state $\rho = \ketbra{\psi}{\psi}$ in the subspace generated by $\ket{x_1}, \ket{x_d}$, for example
\begin{equation}
    \ket{\psi}
    = \cos\left(\frac{\tau_0}{2}\right) \ket{x_d}
    - i\sin\left(\frac{\tau_0}{2}\right) \ket{x_1},
\end{equation}
for some angle $\tau_0$ in $[0, 2\pi]$.
Suppose also that the system evolves according to the Hamiltonian
\begin{equation}
    H = \frac{\hbar\omega}{2} \Big(
        \ketbra{x_1}{x_d}
    + \ketbra{x_d}{x_1}\Big),
\end{equation}
so that the system never leaves the subspace spanned by $\{ \ketbra{x_1}{x_1}, \ketbra{x_d}{x_d} \}$.
In this subspace, the unitary operator can be represented as
\begin{equation}
    U \rightarrow
    \begin{pmatrix}
        \phantom{-i}\cos(\omega t/2) & -i\sin(\omega t/2) \\
        -i\sin(\omega t/2) & \phantom{-i}\cos(\omega t/2) \\
    \end{pmatrix}
\end{equation}
that leads to
\begin{equation}
    \expval{X}_{t}
    =
    x_d \cos^2\left(\frac{\tau_0 + \omega t}{2}\right)
    +
    x_1 \sin^2\left(\frac{\tau_0 + \omega t}{2}\right).
    \label{eq:expval_saturation}
\end{equation}
Eq.~\eqref{eq:expval_saturation} matches either the bound \eqref{eq:lower_bound_general} or \eqref{eq:upper_bound_general} depending on the value of $\tau_0$. If $\tau_0 \in [0, \pi]$, then the expectation value in Eq.~\eqref{eq:expval_saturation} coincides with the time-varying branch of the lower bound \eqref{eq:lower_bound_general}.
On the other hand, if $\tau_0 \in [\pi, 2\pi]$, then $\cos^2((\tau_0 + \omega t)/2) = \cos^2((\tau_0 - \omega t)/2)$ such that Eq.~\eqref{eq:expval_saturation} saturates the upper bound \eqref{eq:upper_bound_general}.
Notice that this reasoning is true because we have assumed that $\rho$ is a pure state, which means: $\Delta L = 2 \Delta H / \hbar = \omega$.

The saturation of the bounds for $\expval{X}_{t}$ entails also the saturation of the bounds on the evolution of the real and imaginary parts of KDQ. In fact, let us consider, for example, the real part of the KDQ $q_{1,1}(t)$ in Fig.~\ref{fig:diffeq_solution}. In this case, the observable $X$ is the operator $\rho_1 = \{\rho, A_1\} / 2$, where $\rho$ is the density operator associated to the pure quantum state, and $A_1$ is a rank-1 projector. We denote with $r_1 \leq \ldots \leq r_d$ the eigenvalues of $\rho_1$, and with $\ket{r_1}, \ldots, \ket{r_d}$ the corresponding eigenvectors, respectively. All the other eigenvalues besides $r_1$ and $r_d$ are zero, in this case. Then, we consider the Hamiltonian $H = - \frac{\hbar \omega}{2} \, (\ketbra{r_1}{r_d} + \ketbra{r_d}{r_1})$, and the operator $B_1 =  \ketbra{\psi}{\psi}$ that plays the role of the operator $\rho$ with respect to which $\expval{X}_{t}$ is computed. In $B_1$, $\ket{\psi} = \cos(\frac{\tau_0}{2}) \ket{r_1} - i \sin(\frac{\tau_0}{2}) \ket{r_d}$ for some angle $\tau_0 \in [0, 2\pi]$. Using the same reasoning of the previous paragraph, we can determine that
\begin{equation}
    \Re{q_{1,1}(t)}
    = \tr{\rho_1 \, U^\dagger B_1 U}
    \equiv
    \expval{\rho_1}_{1,t}
\end{equation}
saturates one of the bounds among \eqref{eq:lower_bound_general} and \eqref{eq:upper_bound_general} for a specific value of $\tau_0$, as illustrated in Fig.~\ref{fig:diffeq_solution}.

\printbibliography

@article{carabba2022quantum,
  title={Quantum speed limits on operator flows and correlation functions},
  author={Carabba, N. and H{\"o}rnedal, N. and del Campo, A.},
  journal={Quantum},
  volume={6},
  pages={884},
  year={2022},
  publisher={Verein zur F{\"o}rderung des Open Access Publizierens in den Quantenwissenschaften},
  doi={10.22331/q-2022-12-22-884},
  url={https://quantum-journal.org/papers/q-2022-12-22-884}
}

@article{Pandey2023PRA,
  title = {Speed limits on correlations in bipartite quantum systems},
  author = {Pandey, V. and Shrimali, D. and Mohan, B. and Das, S. and Pati, A. K.},
  journal = {Phys. Rev. A},
  volume = {107},
  issue = {5},
  pages = {052419},
  numpages = {17},
  year = {2023},
  month = 5,
  publisher = {American Physical Society},
  doi = {10.1103/PhysRevA.107.052419},
  url = {https://link.aps.org/doi/10.1103/PhysRevA.107.052419}
}

@article{Aifer2023PRXQ,
  title = {Thermodynamics of Quantum Information in Noisy Polarizers},
  author = {Aifer, M. and Myers, N. M. and Deffner, S.},
  journal = {PRX Quantum},
  volume = {4},
  issue = {2},
  pages = {020343},
  numpages = {21},
  year = {2023},
  month = 6,
  publisher = {American Physical Society},
  doi = {10.1103/PRXQuantum.4.020343},
  url = {https://link.aps.org/doi/10.1103/PRXQuantum.4.020343}
}

@article{Aifer2024PRL,
  title = {{Energetic Cost for Speedy Synchronization in Non-Hermitian Quantum Dynamics}},
  author = {Aifer, M. and Thingna, J. and Deffner, S.},
  journal = {Phys. Rev. Lett.},
  volume = {133},
  issue = {2},
  pages = {020401},
  numpages = {8},
  year = {2024},
  month = 7,
  publisher = {American Physical Society},
  doi = {10.1103/PhysRevLett.133.020401},
  url = {https://link.aps.org/doi/10.1103/PhysRevLett.133.020401}
}

@article{carolan2024operator,
  title={Operator growth and spread complexity in open quantum systems},
  author={Carolan, E. and Kiely, A. and Campbell, S. and Deffner, S.},
  journal={EPL (Europhys. Lett.)},
  volume={147},
  pages={38002},
  year={2024},
  doi={10.1209/0295-5075/ad5b17},
  url={https://iopscience.iop.org/article/10.1209/0295-5075/ad5b17}
}

@article{Touil2024EPL,
doi = {10.1209/0295-5075/ad4413},
url = {https://dx.doi.org/10.1209/0295-5075/ad4413},
year = {2024},
month = 6,
publisher = {EDP Sciences, IOP Publishing and Società Italiana di Fisica},
volume = {146},
number = {4},
pages = {48001},
author = {Touil, A. and Deffner, S.},
title = {Information scrambling -- A quantum thermodynamic perspective},
journal = {EPL (Europhys. Lett.)}
}

@article{Ollivier2001PRA,
  title = {Quantum Discord: A Measure of the Quantumness of Correlations},
  author = {Ollivier, H. and Zurek, W. H.},
  journal = {Phys. Rev. Lett.},
  volume = {88},
  issue = {1},
  pages = {017901},
  numpages = {4},
  year = {2001},
  month = 12,
  publisher = {American Physical Society},
  doi = {10.1103/PhysRevLett.88.017901},
  url = {https://link.aps.org/doi/10.1103/PhysRevLett.88.017901}
}

@article{Leggett1985PRL,
  title = {Quantum mechanics versus macroscopic realism: Is the flux there when nobody looks?},
  author = {Leggett, A. J. and Garg, A.},
  journal = {Phys. Rev. Lett.},
  volume = {54},
  issue = {9},
  pages = {857--860},
  numpages = {0},
  year = {1985},
  month = 3,
  publisher = {American Physical Society},
  doi = {10.1103/PhysRevLett.54.857},
  url = {https://link.aps.org/doi/10.1103/PhysRevLett.54.857}
}

@article{mohan2022quantum,
  title = {Quantum speed limits for observables},
  author = {Mohan, B. and Pati, A. K.},
  journal = {Phys. Rev. A},
  volume = {106},
  issue = {4},
  pages = {042436},
  numpages = {12},
  year = {2022},
  month = 10,
  publisher = {American Physical Society},
  doi = {10.1103/PhysRevA.106.042436},
  url = {https://link.aps.org/doi/10.1103/PhysRevA.106.042436}
}

@article{LostaglioKirkwood2022,
author = {Lostaglio, M. and Belenchia, A. and Levy, A. and Hern\'andez-G\'omez, S. and Fabbri, N. and Gherardini, S.},
title = {{Kirkwood-Dirac quasiprobability approach to the statistics of incompatible observables}},
journal = {Quantum},
volume = {7},
pages = {1128},
year = {2023},
doi = {10.22331/q-2023-10-09-1128},
url = {https://quantum-journal.org/papers/q-2023-10-09-1128/}
}

@incollection{mandelstam1991uncertainty,
  title={The uncertainty relation between energy and time in non-relativistic quantum mechanics},
  author={Mandelstam, L. and Tamm, I. G.},
  booktitle={Selected papers},
  pages={115--123},
  year={1991},
  publisher={Springer}
}

@article{RobertsonPR1929,
  title = {{The Uncertainty Principle}},
  author = {Robertson, H. P.},
  journal = {Phys. Rev.},
  volume = {34},
  issue = {1},
  pages = {163--164},
  numpages = {0},
  year = {1929},
  month = 7,
  publisher = {American Physical Society},
  doi = {10.1103/PhysRev.34.163},
  url = {https://link.aps.org/doi/10.1103/PhysRev.34.163}
}

@article{schrodinger1999heisenberg,
  title={About {H}eisenberg uncertainty relation},
  author={Schr{\"o}dinger, E.},
  journal={arXiv preprint quant-ph/9903100},
  year={1999},
  doi = {10.48550/arXiv.quant-ph/9903100},
 url = {https://doi.org/10.48550/arXiv.quant-ph/9903100}
}

@article{hernandez2022experimental,
  title = {Projective measurements can probe nonclassical work extraction and time correlations},
  author = {Hern\'andez-G\'omez, S. and Gherardini, S. and Belenchia, A. and Lostaglio, M. and Levy, A. and Fabbri, N.},
  journal = {Phys. Rev. Res.},
  volume = {6},
  issue = {2},
  pages = {023280},
  numpages = {10},
  year = {2024},
  month = 6,
  publisher = {American Physical Society},
  doi = {10.1103/PhysRevResearch.6.023280},
  url = {https://link.aps.org/doi/10.1103/PhysRevResearch.6.023280}
}

@article{levy2020quasiprobability,
  title = {{Quasiprobability Distribution for Heat Fluctuations in the Quantum Regime}},
  author = {Levy, A. and Lostaglio, M.},
  journal = {PRX Quantum},
  volume = {1},
  issue = {1},
  pages = {010309},
  numpages = {19},
  year = {2020},
  month = 9,
  publisher = {American Physical Society},
  doi = {10.1103/PRXQuantum.1.010309},
  url = {https://link.aps.org/doi/10.1103/PRXQuantum.1.010309}
}

@article{CimininpjQI2020,
  title = {{Experimental characterization of the energetics of quantum logic gates}},
  author = {Cimini, V. and Gherardini, S. and Barbieri, M. and Gianani, I. and Sbroscia, M. and Buffoni, L. and Paternostro, M. and Caruso, F.},
  journal = {Npj Quantum Inf.},
  volume = {6},
  issue = {1},
  pages = {96},
  year = {2020},
  doi = {10.1038/s41534-020-00325-7},
  url = {https://www.nature.com/articles/s41534-020-00325-7}
}

@article{ArvidssonShukurJPA2021,
  title = {Conditions tighter than noncommutation needed for nonclassicality},
  author = {Arvidsson-Shukur, D. R. M. and Chevalier Drori, J. and Yunger Halpern, N.},
  journal = {J. Phys. A: Math. Theor.},
  volume = {54},
  issue = {28},
  pages = {284001},
  year = {2021},
  doi = {10.1088/1751-8121/ac0289},
  url = {https://iopscience.iop.org/article/10.1088/1751-8121/ac0289}
}

@book{nielsen2010quantum,
  title={Quantum computation and quantum information},
  author={Nielsen, M. A. and Chuang, I. L.},
  year={2010},
  publisher={Cambridge university press}
}

@article{Touil2020QST,
doi = {10.1088/2058-9565/ab8ebb},
url = {https://dx.doi.org/10.1088/2058-9565/ab8ebb},
year = {2020},
month = 5,
publisher = {IOP Publishing},
volume = {5},
number = {3},
pages = {035005},
author = {Touil, A. and Deffner, S.},
title = {Quantum scrambling and the growth of mutual information},
journal = {Quantum Sci. Technol.}
}

@article{Touil2021PRXQ,
  title = {Information Scrambling versus Decoherence---Two Competing Sinks for Entropy},
  author = {Touil, A. and Deffner, S.},
  journal = {PRX Quantum},
  volume = {2},
  issue = {1},
  pages = {010306},
  numpages = {15},
  year = {2021},
  month = 1,
  publisher = {American Physical Society},
  doi = {10.1103/PRXQuantum.2.010306},
  url = {https://link.aps.org/doi/10.1103/PRXQuantum.2.010306}
}

@article{tripathy2024quantum,
  title={{Quantum information scrambling in two-dimensional Bose-Hubbard lattices}},
  author={Tripathy, D. and Touil, A. and Gardas, B. and Deffner, S.},
  journal={Chaos},
  volume={34},
  number={4},
  pages={043121},
  year={2024},
  publisher={AIP Publishing},
  doi={10.1063/5.0199335}
}

@article{śmierzchalski2023efficiency,
  title={Efficiency optimization in quantum computing: balancing thermodynamics and computational performance},
  author={Śmierzchalski, T. and Mzaouali, Z. and Deffner, S. and Gardas, B.},
  journal={Sci. Rep.},
  volume={14},
  number={4555},
  year={2024},
  publisher={Nature Publishing Group},
  doi={10.1038/s41598-024-55314-z},
  url={https://www.nature.com/articles/s41598-024-55314-z}
}

@article{Aifer2022NJP,
doi = {10.1088/1367-2630/ac6821},
url = {https://dx.doi.org/10.1088/1367-2630/ac6821},
year = {2022},
month = 5,
publisher = {IOP Publishing},
volume = {24},
number = {5},
pages = {055002},
author = {Aifer, M. and Deffner, S.},
title = {From quantum speed limits to energy-efficient quantum gates},
journal = {New J. Phys.}
}

@article{Touil2022JPA,
doi = {10.1088/1751-8121/ac3eba},
url = {https://dx.doi.org/10.1088/1751-8121/ac3eba},
year = {2021},
month = 12,
publisher = {IOP Publishing},
volume = {55},
number = {2},
pages = {025301},
author = {Touil, A. and Çakmak, B. and Deffner, S.},
title = {Ergotropy from quantum and classical correlations},
journal = {J. Phys. A: Math. Theor.}
}

@article{Kafri2012PRA,
  title = {Holevo's bound from a general quantum fluctuation theorem},
  author = {Kafri, D. and Deffner, S.},
  journal = {Phys. Rev. A},
  volume = {86},
  issue = {4},
  pages = {044302},
  numpages = {5},
  year = {2012},
  month = 10,
  publisher = {American Physical Society},
  doi = {10.1103/PhysRevA.86.044302},
  url = {https://link.aps.org/doi/10.1103/PhysRevA.86.044302}
}

@article{Deffner2021EPL,
doi = {10.1209/0295-5075/134/40002},
url = {https://dx.doi.org/10.1209/0295-5075/134/40002},
year = {2021},
month = 7,
publisher = {EDP Sciences, IOP Publishing and Società Italiana di Fisica},
volume = {134},
number = {4},
pages = {40002},
author = {Deffner, S.},
title = {Energetic cost of Hamiltonian quantum gates},
journal = {EPL (Europhys. Lett.)}
}

@article{Buscemi_2014,
	doi = {10.1142/s0219749915600023},
	year = {2014},
	month = 11,
  	publisher = {World Scientific Pub Co Pte Lt},
  	volume = {12},
  	number = {07n08},
  	pages = {1560002},
  	author = {Buscemi, F. and Dall{\textquotesingle}Arno, M. and Ozawa, M. and Vedral, V.},
  	title = {Universal optimal quantum correlator},
  	journal = {Int. J. Quantum Inf.}
}

@article{diaz2020quantum,
  title={Quantum work statistics with initial coherence},
  author={D{\'\i}az, M. G. and Guarnieri, G. and Paternostro, M.},
  journal={Entropy},
  volume={22},
  number={11},
  pages={1223},
  year={2020},
  publisher={MDPI},
  doi={10.3390/e22111223},
  url={https://www.mdpi.com/1099-4300/22/11/1223}
}

@article{maffei2022anomalous,
  title = {{Anomalous energy exchanges and Wigner-function negativities in a single-qubit gate}},
  author = {Maffei, M. and Elouard, C. and Goes, B. O. and Huard, B. and Jordan, A. N. and Auffèves, A.},
  journal = {Phys. Rev. A},
  volume = {107},
  issue = {2},
  pages = {023710},
  numpages = {8},
  year = {2023},
  month = 2,
  publisher = {American Physical Society},
  doi = {10.1103/PhysRevA.107.023710},
  url = {https://link.aps.org/doi/10.1103/PhysRevA.107.023710}
}

@article{SpekkensPRL2008,
  title = {{Negativity and Contextuality are Equivalent Notions of Nonclassicality}},
  author = {Spekkens, R. W.},
  journal = {Phys. Rev. Lett.},
  volume = {101},
  issue = {2},
  pages = {020401},
  numpages = {4},
  year = {2008},
  month = 7,
  publisher = {American Physical Society},
  doi = {10.1103/PhysRevLett.101.020401},
  url = {https://link.aps.org/doi/10.1103/PhysRevLett.101.020401}
}

@article{DeBievrePRL2021,
  title = {{Complete Incompatibility, Support Uncertainty, and Kirkwood-Dirac Nonclassicality}},
  author = {De Bièvre, S.},
  journal = {Phys. Rev. Lett.},
  volume = {127},
  issue = {19},
  pages = {190404},
  numpages = {6},
  year = {2021},
  month = 11,
  publisher = {American Physical Society},
  doi = {10.1103/PhysRevLett.127.190404},
  url = {https://link.aps.org/doi/10.1103/PhysRevLett.127.190404}
}

@article{SantiniPRB2023,
  title = {Work statistics, quantum signatures, and enhanced work extraction in quadratic fermionic models},
  author = {Santini, A. and Solfanelli, A. and Gherardini, S. and Collura, M.},
  journal = {Phys. Rev. B},
  volume = {108},
  issue = {10},
  pages = {104308},
  numpages = {13},
  year = {2023},
  month = 9,
  publisher = {American Physical Society},
  doi = {10.1103/PhysRevB.108.104308},
  url = {https://link.aps.org/doi/10.1103/PhysRevB.108.104308}
}

@article{HofmannPRL2012,
  title = {{How Weak Values Emerge in Joint Measurements on Cloned Quantum Systems}},
  author = {Hofmann, H. F.},
  journal = {Phys. Rev. Lett.},
  volume = {109},
  issue = {2},
  pages = {020408},
  numpages = {5},
  year = {2012},
  month = 7,
  publisher = {American Physical Society},
  doi = {10.1103/PhysRevLett.109.020408},
  url = {https://link.aps.org/doi/10.1103/PhysRevLett.109.020408}
}

@article{alonso2019out,
  title = {{Out-of-Time-Ordered-Correlator Quasiprobabilities Robustly Witness Scrambling}},
  author = {Gonz\'alez Alonso, J. R. and Yunger Halpern, N. and Dressel, J.},
  journal = {Phys. Rev. Lett.},
  volume = {122},
  issue = {4},
  pages = {040404},
  numpages = {7},
  year = {2019},
  month = 2,
  publisher = {American Physical Society},
  doi = {10.1103/PhysRevLett.122.040404},
  url = {https://link.aps.org/doi/10.1103/PhysRevLett.122.040404}
}

@article{SilvaPRL2008,
  title = {{Statistics of the Work Done on a Quantum Critical System by Quenching a Control Parameter}},
  author = {Silva, A.},
  journal = {Phys. Rev. Lett.},
  volume = {101},
  issue = {12},
  pages = {120603},
  numpages = {4},
  year = {2008},
  month = 9,
  publisher = {American Physical Society},
  doi = {10.1103/PhysRevLett.101.120603},
  url = {https://link.aps.org/doi/10.1103/PhysRevLett.101.120603}
}

@article{chenu2018quantum,
  title={Quantum work statistics, {L}oschmidt echo and information scrambling},
  author={Chenu, A. and Egusquiza, I. L. and Molina-Vilaplana, J. and del Campo, A.},
  journal={Sci. Rep.},
  volume={8},
  number={1},
  pages={1--8},
  year={2018},
  publisher={Nature Publishing Group},
  doi={10.1038/s41598-018-30982-w}
}

@article{Dressel2014colloquium,
  title = {{Colloquium: Understanding quantum weak values: Basics and applications}},
  author = {Dressel, J. and Malik, M. and Miatto, F. M. and Jordan, A. N. and Boyd, R. W.},
  journal = {Rev. Mod. Phys.},
  volume = {86},
  issue = {1},
  pages = {307--316},
  numpages = {10},
  year = {2014},
  month = 3,
  publisher = {American Physical Society},
  doi = {10.1103/RevModPhys.86.307},
  url = {https://link.aps.org/doi/10.1103/RevModPhys.86.307}
}

@article{dressel2018strengthening,
  title = {Strengthening weak measurements of qubit out-of-time-order correlators},
  author = {Dressel, J. and Gonz\'alez Alonso, J. R. and Waegell, M. and Yunger Halpern, N.},
  journal = {Phys. Rev. A},
  volume = {98},
  issue = {1},
  pages = {012132},
  numpages = {11},
  year = {2018},
  month = 7,
  publisher = {American Physical Society},
  doi = {10.1103/PhysRevA.98.012132},
  url = {https://link.aps.org/doi/10.1103/PhysRevA.98.012132}
}

@article{Hofer2017quasi,
  title = {Quasi-probability distributions for observables in dynamic systems},
  author = {Hofer, P. P.},
  journal = {Quantum},
  volume = {1},
  pages = {32},
  year = {2017},
  doi = {10.22331/q-2017-10-12-32}
}

@article{johansen2007quantum,
  title = {Quantum theory of successive projective measurements},
  author = {Johansen, L. M.},
  journal = {Phys. Rev. A},
  volume = {76},
  issue = {1},
  pages = {012119},
  numpages = {6},
  year = {2007},
  month = 7,
  publisher = {American Physical Society},
  doi = {10.1103/PhysRevA.76.012119},
  url = {https://link.aps.org/doi/10.1103/PhysRevA.76.012119}
}

@article{KunjwalPRA2019,
  title = {Anomalous weak values and contextuality: {R}obustness, tightness, and imaginary parts},
  author = {Kunjwal, R. and Lostaglio, M. and Pusey, M. F.},
  journal = {Phys. Rev. A},
  volume = {100},
  issue = {4},
  pages = {042116},
  numpages = {19},
  year = {2019},
  month = 10,
  publisher = {American Physical Society},
  doi = {10.1103/PhysRevA.100.042116},
  url = {https://link.aps.org/doi/10.1103/PhysRevA.100.042116}
}

@article{lupu2021negative,
  title = {{Negative Quasiprobabilities Enhance Phase Estimation in Quantum-Optics Experiment}},
  author = {Lupu-Gladstein, N. and Yilmaz, Y. B. and Arvidsson-Shukur, D. R. M. and Brodutch, A. and Pang, A. O. T. and Steinberg, A. M. and Halpern, N. Y.},
  journal = {Phys. Rev. Lett.},
  volume = {128},
  issue = {22},
  pages = {220504},
  numpages = {8},
  year = {2022},
  month = 6,
  publisher = {American Physical Society},
  doi = {10.1103/PhysRevLett.128.220504},
  url = {https://link.aps.org/doi/10.1103/PhysRevLett.128.220504}
}

@article{mohseninia2019strongly,
  title = {Optimizing measurement strengths for qubit quasiprobabilities behind out-of-time-ordered correlators},
  author = {Mohseninia, R. and Gonz\'alez Alonso, J. R. and Dressel, J.},
  journal = {Phys. Rev. A},
  volume = {100},
  issue = {6},
  pages = {062336},
  numpages = {6},
  year = {2019},
  month = 12,
  publisher = {American Physical Society},
  doi = {10.1103/PhysRevA.100.062336},
  url = {https://link.aps.org/doi/10.1103/PhysRevA.100.062336}
}

@article{margenau1961correlation,
  title={Correlation between measurements in quantum theory},
  author={Margenau, H. and Hill, R. N.},
  journal={Prog. Theor. Phys.},
  volume={26},
  number={5},
  pages={722--738},
  year={1961},
  publisher={Oxford University Press},
  doi={10.1143/PTP.26.722}
}

@article{PuseyPRL2014,
  title = {{Anomalous Weak Values Are Proofs of Contextuality}},
  author = {Pusey, M. F.},
  journal = {Phys. Rev. Lett.},
  volume = {113},
  issue = {20},
  pages = {200401},
  numpages = {5},
  year = {2014},
  month = 11,
  publisher = {American Physical Society},
  doi = {10.1103/PhysRevLett.113.200401},
  url = {https://link.aps.org/doi/10.1103/PhysRevLett.113.200401}
}

@article{Solinas2016probing,
  title = {Probing quantum interference effects in the work distribution},
  author = {Solinas, P. and Gasparinetti, S.},
  journal = {Phys. Rev. A},
  volume = {94},
  issue = {5},
  pages = {052103},
  numpages = {11},
  year = {2016},
  month = 11,
  publisher = {American Physical Society},
  doi = {10.1103/PhysRevA.94.052103},
  url = {https://link.aps.org/doi/10.1103/PhysRevA.94.052103}
}

@article{yunger2018quasiprobability,
  title = {Quasiprobability behind the out-of-time-ordered correlator},
  author = {Yunger Halpern, N. and Swingle, B. and Dressel, J.},
  journal = {Phys. Rev. A},
  volume = {97},
  issue = {4},
  pages = {042105},
  numpages = {37},
  year = {2018},
  month = 4,
  publisher = {American Physical Society},
  doi = {10.1103/PhysRevA.97.042105},
  url = {https://link.aps.org/doi/10.1103/PhysRevA.97.042105}
}

@article{PerarnauLlobetPRL2017,
  title = {{No-Go Theorem for the Characterization of Work Fluctuations in Coherent Quantum Systems}},
  author = {Perarnau-Llobet, M. and B\"aumer, E. and Hovhannisyan, K. V. and Huber, M. and Acin, A.},
  journal = {Phys. Rev. Lett.},
  volume = {118},
  issue = {7},
  pages = {070601},
  numpages = {6},
  year = {2017},
  month = 2,
  publisher = {American Physical Society},
  doi = {10.1103/PhysRevLett.118.070601},
  url = {https://link.aps.org/doi/10.1103/PhysRevLett.118.070601}
}

@article{TalknerPRE2007,
  title = {{Fluctuation theorems: Work is not an observable}},
  author = {Talkner, P. and Lutz, E. and H\"anggi, P.},
  journal = {Phys. Rev. E},
  volume = {75},
  issue = {5},
  pages = {050102(R)},
  numpages = {2},
  year = {2007},
  month = 5,
  publisher = {American Physical Society},
  doi = {10.1103/PhysRevE.75.050102},
  url = {https://link.aps.org/doi/10.1103/PhysRevE.75.050102}
}

@article{buffoni2020thermodynamics,
  title={Thermodynamics of a quantum annealer},
  author={Buffoni, L. and Campisi, M.},
  journal={Quantum Sci. Technol.},
  volume={5},
  issue={3},
  pages={035013},
  year={2020},
  doi={10.1088/2058-9565/ab9755},
  url={https://iopscience.iop.org/article/10.1088/2058-9565/ab9755}
}

@article{BuffoniPRL2022,
  title = {Third Law of Thermodynamics and the Scaling of Quantum Computers},
  author = {Buffoni, L. and Gherardini, S. and Zambrini Cruzeiro, E. and Omar, Y.},
  journal = {Phys. Rev. Lett.},
  volume = {129},
  issue = {15},
  pages = {150602},
  numpages = {7},
  year = {2022},
  month = 10,
  publisher = {American Physical Society},
  doi = {10.1103/PhysRevLett.129.150602},
  url = {https://link.aps.org/doi/10.1103/PhysRevLett.129.150602}
}

@article{campisi2011colloquium,
  title = {{Colloquium: Quantum fluctuation relations: Foundations and applications}},
  author = {Campisi, M. and H\"anggi, P. and Talkner, P.},
  journal = {Rev. Mod. Phys.},
  volume = {83},
  issue = {3},
  pages = {771--791},
  numpages = {0},
  year = {2011},
  month = 7,
  publisher = {American Physical Society},
  doi = {10.1103/RevModPhys.83.771},
  url = {https://link.aps.org/doi/10.1103/RevModPhys.83.771}
}

@article{stevens2022energetics,
  title = {{Energetics of a Single Qubit Gate}},
  author = {Stevens, J. and Szombati, D. and Maffei, M. and Elouard, C. and Assouly, R. and Cottet, N. and Dassonneville, R. and Ficheux, Q. and Zeppetzauer, S. and Bienfait, A. and Jordan, A. N. and Auffèves, A. and Huard, B.},
  journal = {Phys. Rev. Lett.},
  volume = {129},
  issue = {11},
  pages = {110601},
  numpages = {7},
  year = {2022},
  month = 9,
  publisher = {American Physical Society},
  doi = {10.1103/PhysRevLett.129.110601},
  url = {https://link.aps.org/doi/10.1103/PhysRevLett.129.110601}
}

@article{Gianani2022diagnostics,
  title={Diagnostics of quantum-gate coherences deteriorated by unitary errors via end-point-measurement statistics},
  author={Gianani, I. and Belenchia, A. and Gherardini, S. and Berardi, V. and Barbieri, M. and Paternostro, M.},
  journal={Quantum Sci. Technol.},
  volume={8},
  number={4},
  pages={045018},
  year={2023},
  publisher={IOP Publishing},
  doi = {https://doi.org/10.1088/2058-9565/acedca},
  url = {https://doi.org/10.1088/2058-9565/acedca}
}

@article{SolinasPRAmeasurement,
  title = {Measurement of work and heat in the classical and quantum regimes},
  author = {Solinas, P. and Amico, M. and Zanghì, N.},
  journal = {Phys. Rev. A},
  volume = {103},
  issue = {6},
  pages = {L060202},
  numpages = {5},
  year = {2021},
  month = 6,
  publisher = {American Physical Society},
  doi = {10.1103/PhysRevA.103.L060202},
  url = {https://link.aps.org/doi/10.1103/PhysRevA.103.L060202}
}

@article{wagner2023quantum,
author={Wagner, R. and Schwartzman-Nowik, Z. and Paiva, I. L. and Te'eni, A. and Ruiz-Molero, A. and Barbosa, R. Soares and Cohen, E. and Galv{\~a}o, E. F.},
title = {{Quantum circuits for measuring weak values, Kirkwood–Dirac quasiprobability distributions, and state spectra}},
journal = {Quantum Sci. Technol.},
volume = {9},
issue = {1},
pages = {015030},
year = {2024},
doi = {10.1088/2058-9565/ad124c},
url = {https://iopscience.iop.org/article/10.1088/2058-9565/ad124c}
}

@article{BudiyonoPRAquantifying,
  title = {{Quantifying quantum coherence via Kirkwood-Dirac quasiprobability}},
  author = {Budiyono, A. and Dipojono, H. K.},
  journal = {Phys. Rev. A},
  volume = {107},
  issue = {2},
  pages = {022408},
  numpages = {9},
  year = {2023},
  month = 2,
  publisher = {American Physical Society},
  doi = {10.1103/PhysRevA.107.022408},
  url = {https://link.aps.org/doi/10.1103/PhysRevA.107.022408}
}

@article{HalliwellPRA2016,
  title = {{Leggett-Garg inequalities and no-signaling in time: A quasiprobability approach}},
  author = {Halliwell, J. J.},
  journal = {Phys. Rev. A},
  volume = {93},
  issue = {2},
  pages = {022123},
  numpages = {9},
  year = {2016},
  month = 2,
  publisher = {American Physical Society},
  doi = {10.1103/PhysRevA.93.022123},
  url = {https://link.aps.org/doi/10.1103/PhysRevA.93.022123}
}

@article{MandelstamJPhys1945,
author = {Mandelstam, L. I. and Tamm, I. E.},
title = {{The uncertainty relation between energy and time in nonrelativistic quantum mechanics}},
journal = {J. Phys.},
volume = {9},
pages = {249--254},
year = {1945}
}

@article{GardasSciRep2018,
author = {Gardas, B. and Deffner, S.},
title = {{Quantum fluctuation theorem for error diagnostics in quantum annealers}},
journal = {Sci. Rep.},
volume = {8},
pages = {17191},
year = {2018},
doi = {10.1038/s41598-018-35264-z},
url = {https://www.nature.com/articles/s41598-018-35264-z}
}

@article{DeffnerJPA2017,
author = {Deffner, S. and Campbell, S.},
title = {{Quantum speed limits: from Heisenberg's uncertainty principle to optimal quantum control}},
journal = {J. Phys. A: Math. Theor.},
volume = {50},
issue = {45},
pages = {453001},
year = {2017},
doi = {10.1088/1751-8121/aa86c6},
url = {https://iopscience.iop.org/article/10.1088/1751-8121/aa86c6}
}

@article{FrancicaPRE2023,
  title = {{Quasiprobability distribution of work in the quantum Ising model}},
  author = {Francica, G. and Dell'Anna, L.},
  journal = {Phys. Rev. E},
  volume = {108},
  issue = {1},
  pages = {014106},
  numpages = {15},
  year = {2023},
  month = 7,
  publisher = {American Physical Society},
  doi = {10.1103/PhysRevE.108.014106},
  url = {https://link.aps.org/doi/10.1103/PhysRevE.108.014106}
}

@article{PeiPRE2023,
  title = {{Exploring quasiprobability approaches to quantum work in the presence of initial coherence: Advantages of the Margenau-Hill distribution}},
  author = {Pei, J.-H. and Chen, J.-F. and Quan, H. T.},
  journal = {Phys. Rev. E},
  volume = {108},
  issue = {5},
  pages = {054109},
  numpages = {14},
  year = {2023},
  month = 11,
  publisher = {American Physical Society},
  doi = {10.1103/PhysRevE.108.054109},
  url = {https://link.aps.org/doi/10.1103/PhysRevE.108.054109}
}

@article{UzdinPRX2015,
  title = {{Equivalence of Quantum Heat Machines, and Quantum-Thermodynamic Signatures}},
  author = {Uzdin, R. and Levy, A. and Kosloff, R.},
  journal = {Phys. Rev. X},
  volume = {5},
  issue = {3},
  pages = {031044},
  numpages = {21},
  year = {2015},
  month = 9,
  publisher = {American Physical Society},
  doi = {10.1103/PhysRevX.5.031044},
  url = {https://link.aps.org/doi/10.1103/PhysRevX.5.031044}
}

@article{MyersAVS2022,
author={Myers, N. M. and Abah, O. and Deffner, S.},
title = {{Quantum thermodynamic devices: From theoretical proposals to experimental reality}},
journal = {AVS Quantum Sci.},
volume = {4},
issue = {2},
pages = {027101},
year = {2022},
doi = {10.1116/5.0083192},
url = {https://pubs.aip.org/avs/aqs/article/4/2/027101/2835276}
}

@article{CangemiArXiv2023,
  title = {{Quantum Engines and Refrigerators}},
  author={Cangemi, L. M. and Bhadra, C. and Levy, A.},
  journal = {Phys. Rep.},
  volume = {1087},
  pages = {1--71},
  year = {2024},
  doi = {10.1016/j.physrep.2024.07.001},
  url = {https://www.sciencedirect.com/science/article/pii/S0370157324002710}
}

@article{ReviewQBattery,
  title = {Colloquium: Quantum batteries},
  author = {Campaioli, F. and Gherardini, S. and Quach, J. Q. and Polini, M. and Andolina, G. M.},
  journal = {Rev. Mod. Phys.},
  volume = {96},
  issue = {3},
  pages = {031001},
  numpages = {30},
  year = {2024},
  month = 7,
  publisher = {American Physical Society},
  doi = {10.1103/RevModPhys.96.031001},
  url = {https://link.aps.org/doi/10.1103/RevModPhys.96.031001}
}

@article{FedorovReview2022,
author={Fedorov, A.K. and Gisin, N. and Beloussov, S.M. and Lvovsky, A.I.},
title = {{Quantum computing at the quantum advantage threshold: a down-to-business review}},
journal = {arXiv preprint arXiv:2203.17181},
year = {2022},
doi = {10.48550/arXiv.2203.17181},
url = {https://arxiv.org/abs/2203.17181}
}

@article{GherardiniTutorial,
  title = {{Quasiprobabilities in Quantum Thermodynamics and Many-Body Systems}},
  author = {Gherardini, S. and De Chiara, G.},
  journal = {PRX Quantum},
  volume = {5},
  issue = {3},
  pages = {030201},
  numpages = {38},
  year = {2024},
  month = 9,
  publisher = {American Physical Society},
  doi = {10.1103/PRXQuantum.5.030201},
  url = {https://link.aps.org/doi/10.1103/PRXQuantum.5.030201}
}

@article{ArvidssonShukur2024review,
title={Properties and applications of the Kirkwood-Dirac distribution},
author={Arvidsson-Shukur, D. R. M. and Braasch Jr., W. F. and De Bievre, S. and Dressel, J. and Jordan, A. N. and Langrenez, C. and Lostaglio, M. and Lundeen, J. S. and Yunger Halpern, N.},
journal={New J. Phys.},
volume={26},
pages={121201},
year={2024},
doi = {10.1088/1367-2630/ada05d},
url = {https://iopscience.iop.org/article/10.1088/1367-2630/ada05d}
}

@article{hernandezArXiv2024Interfero,
  title={Interferometry of quantum correlation functions to access quasiprobability distribution of work},
  author={Hern{\'a}ndez-G{\'o}mez, S. and Isogawa, T. and Belenchia, A. and Levy, A. and Fabbri, N. and Gherardini, S. and Cappellaro, P.},
  journal = {Npj Quantum Inf.},
  volume = {10},
  pages = {115},
  year = {2024},
  doi = {10.1038/s41534-024-00913-x},
  url = {https://www.nature.com/articles/s41534-024-00913-x}
}

@article{PiacentiniPRL2016,
  title = {{Experiment Investigating the Connection between Weak Values and Contextuality}},
  author = {Piacentini, F. and Avella, A. and Levi, M. P. and Lussana, R. and Villa, F. and Tosi, A. and Zappa, F. and Gramegna, M. and Brida, G. and Degiovanni, I. P. and Genovese, M.},
  journal = {Phys. Rev. Lett.},
  volume = {116},
  issue = {18},
  pages = {180401},
  numpages = {5},
  year = {2016},
  month = 5,
  publisher = {American Physical Society},
  doi = {10.1103/PhysRevLett.116.180401},
  url = {https://link.aps.org/doi/10.1103/PhysRevLett.116.180401}
}

@article{piacentini2016measuring,
  title = {{Measuring Incompatible Observables by Exploiting Sequential Weak Values}},
  author = {Piacentini, F. and Avella, A. and Levi, M. P. and Gramegna, M. and Brida, G. and Degiovanni, I. P. and Cohen, E. and Lussana, R. and Villa, F. and Tosi, A. and Zappa, F. and Genovese, M.},
  journal = {Phys. Rev. Lett.},
  volume = {117},
  issue = {17},
  pages = {170402},
  numpages = {6},
  year = {2016},
  month = 10,
  publisher = {American Physical Society},
  doi = {10.1103/PhysRevLett.117.170402},
  url = {https://link.aps.org/doi/10.1103/PhysRevLett.117.170402}
}

@article{CiminiQST2020,
  title = {{Anomalous values, Fisher information, and contextuality, in generalized quantum measurements}},
  author = {Cimini, V. and Gianani, I. and Piacentini, F. and Degiovanni, I. P. and Barbieri, M.},
  journal = {Quantum Sci. Technol.},
  volume = {5},
  issue = {2},
  pages = {025007},
  year = {2020},
  publisher = {IOP Publishing},
  doi = {10.1088/2058-9565/ab7988},
  url = {https://dx.doi.org/10.1088/2058-9565/ab7988}
}

@article{HePRA2024,
  title = {{Nonclassicality of the Kirkwood-Dirac quasiprobability distribution via quantum modification terms}},
  author = {He, J. and Fu, S.},
  journal = {Phys. Rev. A},
  volume = {109},
  issue = {1},
  pages = {012215},
  numpages = {12},
  year = {2024},
  month = 1,
  publisher = {American Physical Society},
  doi = {10.1103/PhysRevA.109.012215},
  url = {https://link.aps.org/doi/10.1103/PhysRevA.109.012215}
}

@article{VeitchNJP2012,
  title = {{Negative quasi-probability as a resource for quantum computation}},
  author = {Veitch, V. and Ferrie, C. and Gross, D. and Emerson, J.},
  journal = {New J. Phys.},
  volume = {14},
  pages = {113011},
  year = {2012},
  publisher = {IOP Publishing},
  doi = {10.1088/1367-2630/14/11/113011},
  url = {https://iopscience.iop.org/article/10.1088/1367-2630/14/11/113011}
}

@article{PhysRevX.12.011038,
  title = {Unifying Quantum and Classical Speed Limits on Observables},
  author = {Garc\'{\i}a-Pintos, L. P. and Nicholson, S. B. and Green, J. R. and del Campo, A. and Gorshkov, A. V.},
  journal = {Phys. Rev. X},
  volume = {12},
  issue = {1},
  pages = {011038},
  numpages = {22},
  year = {2022},
  month = 2,
  publisher = {American Physical Society},
  doi = {10.1103/PhysRevX.12.011038},
  url = {https://link.aps.org/doi/10.1103/PhysRevX.12.011038}
}

@article{ShrimaliPRA2024,
  title = {{Stronger speed limit for observables: Tighter bound for the capacity of entanglement, the modular Hamiltonian, and the charging of a quantum battery}},
  author = {Shrimali, D. and Panda, B. and Pati, A. K.},
  journal = {Phys. Rev. A},
  volume = {110},
  issue = {2},
  pages = {022425},
  numpages = {14},
  year = {2024},
  month = 8,
  publisher = {American Physical Society},
  doi = {10.1103/PhysRevA.110.022425},
  url = {https://link.aps.org/doi/10.1103/PhysRevA.110.022425}
}

@article{nicholson2020time,
  title={Time--information uncertainty relations in thermodynamics},
  author={Nicholson, S. B. and Garc{\'\i}a-Pintos, L. P. and del Campo, A. and Green, J. R.},
  journal={Nature Physics},
  volume={16},
  number={12},
  pages={1211--1215},
  year={2020},
  publisher={Nature Publishing Group UK London},
  url={https://www.nature.com/articles/s41567-020-0981-y},
  doi={10.1038/s41567-020-0981-y}
}

@article{SchmidPRA2024,
  title = {{Kirkwood-Dirac representations beyond quantum states and their relation to noncontextuality}},
  author = {Schmid, D. and Baldij\~ao, R. D. and Y\ifmmode \bar{\imath}\else \={\i}\fi{}ng, Y. and Wagner, R. and Selby, J. H.},
  journal = {Phys. Rev. A},
  volume = {110},
  issue = {5},
  pages = {052206},
  numpages = {12},
  year = {2024},
  month = {Nov},
  publisher = {American Physical Society},
  doi = {10.1103/PhysRevA.110.052206},
  url = {https://link.aps.org/doi/10.1103/PhysRevA.110.052206}
}

@article{SchmidQuantum2024,
  author = {Schmid, D. and Selby, J. H. and Pusey, M. F. and Spekkens, R. W.},
  title = {A structure theorem for generalized-noncontextual ontological models},
  journal = {Quantum},
  volume = {8},
  pages = {1283},
  year = {2024},
  doi = {10.22331/q-2024-03-14-1283},
  url = {https://doi.org/10.22331/q-2024-03-14-1283}
}

@article{pusz1978passive,
  author = {Pusz, W. and Woronowicz, S. L.},
  title = {{Passive states and KMS states for general quantum systems}},
  journal = {Commun. Math. Phys.},
  volume = {58},
  year = {1978},
  pages = {273--290}
}

@article{Alicki1979,
  author = {Alicki, R.},
  title = {The quantum open system as a model of the heat engine},
  journal = {J. Phys. A: Math. Gen.},
  volume = {12},
  issue = {5},
  year = {1979},
  pages = {L103}
}

\end{document}